\documentclass{statsoc}

\usepackage{amsmath, amssymb, amsthm}
\usepackage{amsfonts,multirow}
\usepackage{graphicx,verbatim}
\usepackage{setspace,color}
\usepackage[normalem]{ulem}
\usepackage[left=2.3cm,top=2.3cm,right=2.3cm,bottom=2.3cm]{geometry}

\usepackage{amsmath,amssymb}
\usepackage{epsfig,psfrag}
\usepackage{graphicx}
\usepackage{verbatim}

\newtheorem{lemma}{Lemma}[section] %%    with section number.

\newtheorem{definition}{Definition}[section]
\newcommand{\bed}{\begin{definition}}
\newcommand{\minimize}{\operatornamewithlimits{minimize}}
\newcommand{\eed}{\end{definition}}
\newcommand{\beq}{\begin{equation}}
\newcommand{\eeq}{\end{equation}}

\newtheorem{theorem}{Theorem}

\newcommand{\bitem}{\begin{itemize}}
\newcommand{\eitem}{\end{itemize}}

\newcommand{\maximize}{\mathrm{maximize}}
\newcommand{\trace}{\mathrm{trace}}

\newcommand{\beqn}{\begin{equation}}
\newcommand{\eeqn}{\end{equation}}
\newcommand{\balign}{\begin{align}}
\newcommand{\ealign}{\end{align}}

\newcommand{\sgn}{\mathrm{sgn}}

\title[Joint graphical lasso for estimation of multiple networks]{The joint graphical lasso for inverse covariance estimation across multiple classes}

\author{Patrick Danaher }
\address{Department of   Biostatistics,  University of Washington,  USA}
\author{Pei Wang \footnote{Corresponding Author. Email: pwang@fhcrc.org. Address: 1100 Fairview Avenue N., M2-B230, Seattle WA 98101.}}
\address{Public Health Sciences Division, Fred Hutchinson Cancer Research Center, USA}
\author[Danaher et al.]{Daniela M. Witten}
\address{Department of Biostatistics, University of Washington, USA}

\begin{document}
%\title{The joint graphical lasso for inverse covariance estimation across multiple classes}
%\author{Patrick Danaher, Pei Wang and Daniela Witten}
%\maketitle

%\begin{quote}
%{\bf Keywords}:
%\end{quote}
%\begin{quote} \small
%{\bf AMS 2000 subject classifications}:   Primary 62G05, 62G10;
%secondary 62G20.
%\end{quote}

\newcommand{\bTheta}{{\bf \Theta}}
\newcommand{\bThetak}{{\bf \Theta}^{(k)}}
\newcommand{\bZk}{{\bf Z}^{(k)}}
\newcommand{\bZK}{{\bf Z}^{(K)}}
\newcommand{\bSk}{{\bf S}^{(k)}}
\newcommand{\bUk}{{\bf U}^{(k)}}
\newcommand{\bUK}{{\bf U}^{(K)}}
\newcommand{\bAk}{{\bf A}^{(k)}}
\newcommand{\bThetahat}{ \hat{\bTheta} }
\newcommand{\btheta}{\boldsymbol{\theta}}
\newcommand{\siginv}{{\bf \Sigma}^{-1}}
\newcommand{\sTheta}{ \{ \bTheta \}  }
\newcommand{\sZ}{ \{ {\bf Z} \}  }
\newcommand{\sU}{ \{ {\bf U} \}  }
\newcommand{\sThetahat}{ \{ \hat{\bTheta} \}  }
\newcommand{\fp}{ \ell_P }  % the objective function

\begin{abstract}
We consider the problem of estimating multiple related Gaussian graphical models from a high-dimensional data set with observations belonging to distinct classes. We propose the \emph{joint graphical lasso}, which borrows strength across the classes in order to estimate multiple graphical models that share certain characteristics, such as the locations or weights of nonzero edges. Our approach is based upon maximizing a penalized log likelihood. We employ generalized fused lasso or group lasso penalties, and implement a fast ADMM algorithm to solve the corresponding convex optimization problems. The performance of the proposed method is illustrated through simulated and real data examples. \\

{\bf Keywords}:
{alternating directions method of multipliers; generalized fused lasso; group lasso; graphical lasso; network estimation; Gaussian graphical model; high-dimensional }
\end{abstract}

\section{Introduction}

In recent years, much interest has focused upon estimating an undirected graphical model on the basis of a $n \times p$ data
matrix $\bf X$, where $n$ is the number of observations and $p$ is the number of features.
Suppose that the observations  ${\bf x}_1, \ldots, {\bf x}_n \in \mathbb{R}^p$ are independent and identically distributed
$N(\boldsymbol\mu, {\bf \Sigma})$, where $\boldsymbol\mu \in \mathbb{R}^p$ and $\bf\Sigma$ is a positive definite $p \times p$ matrix.
 Then zeros in the inverse covariance matrix ${\bf \Sigma}^{-1}$ correspond to pairs of features
that are conditionally independent -- that is, pairs of variables that are independent of each other, given all of the other
variables in the data set.
In a Gaussian graphical model \citep{Lauritzen1996}, these conditional dependence relationships are represented by a graph in which nodes
 represent features and edges connect conditionally dependent pairs of features.

A natural way to estimate the \emph{precision} (or
\emph{concentration})
  matrix $\siginv$ is via maximum likelihood.
Letting $\bf S$ denote the empirical covariance matrix of $\bf X$, 
the Gaussian log likelihood takes the form (up to a constant)
\begin{equation}
%\maximize_{\bTheta} \left\{ \log \det \bTheta - \trace ({\bf S }\bTheta) \right\}
\frac{n}{2} \left( \log \det {\siginv} - \trace ({\bf S} {\siginv}) \right).
 \label{loglik}
\end{equation}
%where ${\bf \Theta}$ is an optimization variable used to estimate $\siginv$.}
Maximizing (\ref{loglik}) with respect to $\siginv$ yields the maximum likelihood estimate ${\bf S}^{-1}$.

However, two problems can arise in using this maximum likelihood
approach to estimate $\siginv$. First, in the
high-dimensional setting where the number of features $p$ is larger
than the number of observations $n$, the empirical covariance matrix
$\bf S$ is singular and so cannot be inverted to yield an estimate
of $\siginv$. If $p \approx n$, then even if $\bf S$ is not
singular, the maximum likelihood estimate for $\siginv$
 will suffer from very high variance. Second,
 one often is interested in identifying pairs of variables that are unconnected in the graphical model, i.e. that are conditionally
independent; these correspond to zeros in $\siginv$. But maximizing the log likelihood (\ref{loglik}) will in general yield
an estimate of $\siginv$ with no elements that are exactly equal to zero.

In recent years, a number of proposals have been made for estimating
$\siginv$ in the high-dimensional setting in such a way that the
resulting estimate is \emph{sparse}. \citet{MB2006} proposed
doing this via a penalized regression approach, which was extended
by \citet{Space}.  A number of authors have instead taken a
penalized log likelihood approach
\citep{YuanLin07,SparseInv,Rothman08}: rather than maximizing
(\ref{loglik}), one can instead solve the problem
\begin{equation}
\maximize_{\bTheta} \left\{ \log \det \bTheta - \trace ({\bf S }\bTheta) - \lambda ||\bTheta||_1 \right\},
 \label{loglikl1}
\end{equation}
%where again $\bTheta$ is an optimization variable used to estimate $\siginv$, and
where $\lambda$ is a nonnegative tuning parameter. The solution to this optimization problem
provides an estimate for $\siginv$.
The use of an $\ell_1$ or \emph{lasso}  \citep{Ti96} penalty on the elements of  $\bTheta$ 
has the effect that when the tuning parameter
$\lambda$ is large, some elements of the resulting precision matrix estimate will be
exactly equal to zero. Moreover, %unlike (\ref{loglik}),
(\ref{loglikl1}) can be solved even if $p \gg n$. The solution to
the problem (\ref{loglikl1}) is referred to as the \emph{graphical
lasso}. Some authors have proposed applying the $\ell_1$ penalty in (\ref{loglikl1}) only to the off-diagonal elements of $\bf \Theta$.

Graphical models are especially of interest in the analysis of
 gene expression data, since
it is believed that genes operate in pathways, or networks.
Graphical models based on gene expression data can provide a useful
tool for visualizing the relationships among genes and for generating biological hypotheses.
%about which genes are particularly important to the underlying biology.
The standard formulation for estimating a Gaussian graphical model
assumes that each observation is drawn from the same distribution.
However, in many datasets the observations may correspond to several
distinct classes, so the assumption that all observations are drawn
from the same distribution is inappropriate.
%In this paper, we consider
%the problem of Gaussian graphical model estimation in this more realistic setting.
For instance, suppose that a cancer researcher  collects gene expression measurements for a set of cancer tissue samples and a set of normal tissue
samples. In this case, one might want to estimate a graphical model for the cancer samples and a graphical model for the normal samples.
One would expect the two graphical models to be similar to each other, since both are based upon the same type of tissue, but also to have important differences stemming from
 the fact that gene networks are often dysregulated in cancer.
Estimating separate graphical models for the cancer and normal
samples does not exploit the similarity between the true graphical
models. And estimating a single graphical model for the cancer and
normal samples ignores the fact that we do not expect the true
graphical models to be identical, and that the differences between
the graphical models may be of interest.

In this paper, we propose the \emph{joint graphical lasso},
a technique for jointly estimating multiple graphical models corresponding to distinct
but related conditions, such as cancer and normal tissue. Our approach is an extension of the graphical lasso (\ref{loglikl1})
to the case of multiple data sets. It is based upon a penalized log likelihood approach, where the choice of penalty depends on
the characteristics of the graphical models that we expect to be shared across conditions.

We illustrate our method with a small toy example  that consists of
observations from two classes. Within each class, the observations
are independent and identically distributed according to a normal
distribution. The two classes have distinct covariance matrices.
When we apply the graphical lasso separately to the observations in
each class, the resulting graphical model estimates are less
accurate than when we use our joint graphical lasso approach.
Results are shown in Figure \ref{fig:toy}.

\begin{figure}[htp]
\centering
\includegraphics[width=0.5\linewidth]{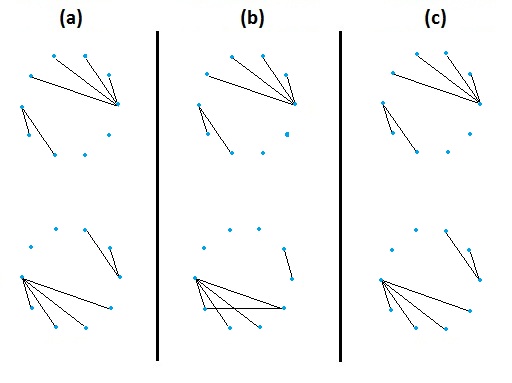}
 \caption{\it{ \label{fig:toy}} \it Comparison of the  graphical lasso with our joint graphical lasso in a toy example with two conditions, $p=10$ variables, and n=200 observations per condition.
{\bf (a)}: True networks. {\bf (b)}: Networks estimated by applying the
graphical lasso separately to each class. {\bf (c)}: Networks estimated by applying our joint graphical lasso proposal.}
\end{figure}

The rest of this paper is organized as follows.
In Section \ref{jointgraphicallasso}, we present the joint graphical lasso optimization problem.
Section \ref{alg} contains an alternating directions method of multipliers algorithm  for its solution.
In Section \ref{screen}, we present theoretical results that lead to massive gains
in the algorithm's computational efficiency.
Section \ref{related} contains a discussion of  related approaches from the literature, and in Section \ref{tuning_parameter_selection}
 we discuss tuning parameter selection.
In Section \ref{simstudy}, we illustrate the performance of our proposal in a simulation study. Section \ref{realstudy} contains an application
to a lung cancer gene expression dataset. The discussion is in Section \ref{discussion}.

\section{The joint graphical lasso}
\label{jointgraphicallasso}

We briefly introduce some notation that will be used throughout this paper.

We let  $K$ denote the number of classes in our data, and let
$\siginv_k$ denote the true precision matrix for the $k$th class. We
will seek to estimate $\siginv_1,\ldots,\siginv_K$ by formulating
convex optimization problems with arguments $ \sTheta =
\bTheta^{(1)}, \ldots, \bTheta^{(K)}$. The solutions
$\hat{\bTheta}^{(1)}, \ldots, \hat{\bTheta}^{(K)}$ to these
optimization problems will constitute estimates of
$\siginv_1,\ldots,\siginv_K$.
% of to be the classes' estimated concentration matrices.  We index classes with $k$, and we index matrix elements with $ij$.

We will index matrix elements using $i=1,\ldots,p$, $j=1,\ldots,p$, and will index classes using $k=1,\ldots,K$.
 $\|{\bf A}\|_F$ will denote the Frobenius norm of matrix ${\bf A}$, \emph{i.e.}
 $\|{\bf A}\|_F = \sqrt{\sum_i \sum_j A_{ij}^2}$ .

\subsection{The general formulation for the joint graphical lasso}
Suppose that we are given $K$ data sets, ${\bf Y}^{(1)}, \ldots, {\bf Y}^{(K)}$, with $K \geq 2$.
 ${\bf Y}^{(k)}$ is a $n_k \times p$ matrix
consisting of  $n_k$ observations with measurements on a set of $p$
features, which are common to all $K$ data sets. Furthermore, we
assume that the $\sum_{k=1}^K n_k$ observations are independent, and
that the observations within each data set are identically
distributed: ${\bf y}_1^{(k)}, \ldots, {\bf y}_{n_k}^{(k)} \sim
N(\boldsymbol\mu_k, {\bf \Sigma}_k)$. Without loss of generality, we
assume that the features within each data set are centered such that
$\boldsymbol\mu_k={\bf 0}$. We let ${\bf S}^{(k)} = \frac{1}{n_k}(
{\bf Y}^{(k)}  )^T {\bf Y}^{(k)} $, the empirical covariance matrix
for ${\bf Y}^{(k)}$.
%We  define $\sTheta = \{\bTheta^{(1)}, \ldots , \bTheta^{(K)}\}$ to be the arguments representing each class's $\siginv$ in the objective function.
%And we call the true inverse covariance matrices $\{({\bf \Sigma}^{(1)})^{-1}, \ldots , ({\bf \Sigma}^{(K)})^{-1} \}$.
The log likelihood for the data takes the form (up to a constant)
\begin{equation}
\ell({\sTheta})=\frac{1}{2} \sum_{k=1}^K  n_k \left(\log \det {\bf \Theta}^{(k)} - \trace( {\bf S}^{(k)}{\bf \Theta}^{(k)}) \right).
\label{loglikjoint}
\end{equation}
Maximizing (\ref{loglikjoint}) with respect to ${\bf \Theta}^{(1)},\ldots,{\bf \Theta}^{(K)}$ yields the maximum likelihood estimate
$({\bf S}^{(1)})^{-1}, \ldots, ({\bf S}^{(K)})^{-1}$.% for $k=1,\ldots,K$.

However, depending on the application, the maximum likelihood
estimates that result from (\ref{loglikjoint}) may not be
satisfactory. When $p$ is smaller than but close to $n_k$, the maximum likelihood estimate can have very high variance, and no elements of $({\bf S}^{(1)})^{-1}, \ldots, ({\bf S}^{(K)})^{-1}$ will be zero, leading to difficulties in interpretation. 
In addition, when $p>n_k$, the maximum likelihood estimate  becomes ill-defined.  
Moreover, if the $K$ data sets correspond to observations collected from $K$
distinct but related classes, then one might wish to borrow strength
across the $K$ classes to estimate the $K$ precision matrices,
rather than estimating each precision matrix separately.

Therefore, instead of estimating $\siginv_1, \ldots, {\siginv_K}$ by
maximizing (\ref{loglikjoint}), we consider the
penalized log likelihood and seek $\sThetahat$ solving
\begin{equation}
%\fp({\sTheta})
\maximize_{\sTheta} \left\{\sum_{k=1}^K  n_k \left(\log \det {\bTheta}^{(k)}
- \trace\left( {\bf S}^{(k)} {\bTheta}^{(k)}\right) \right)- P(\sTheta)
\right\}\label{loglikjointproposal}
\end{equation}
subject to the constraint that ${\bf \Theta}^{(1)}, \ldots, {\bf
\Theta}^{(K)}$ are positive definite. Here  $P(\sTheta)$ denotes a
convex penalty function, so that the objective in
(\ref{loglikjointproposal}) is strictly concave in $\sTheta$.
%When $P(\sTheta)=0$, (\ref{loglikjointproposal}) amounts to estimating $K$ graphical models separately via maximum likelihood.
%We propose instead to choose a penalty function $P$ that
We propose to choose a penalty function $P$ that
will encourage ${\bThetahat}^{(1)}, \ldots, {\bThetahat}^{(K)}$ to share certain characteristics, such as the locations or values of the nonzero elements; moreover, we
would like the estimated precision matrices to be sparse.
In particular, we will consider penalty functions that take the form
$P(\sTheta)=\tilde{P}(\sTheta)+\lambda_1  \sum_k \sum_{i \neq j} |\theta^{(k)}_{ij}|$, where $\tilde{P}$ is a convex function and $\lambda_1$ is a nonnegative tuning parameter.
When $\tilde{P}(\sTheta)=0$, (\ref{loglikjointproposal}) amounts to performing $K$ uncoupled graphical lasso optimization problems (\ref{loglikl1}).
The $\tilde{P}$ penalty is chosen to encourage similarity across the $K$ estimated precision matrices; therefore, we refer to the solution to (\ref{loglikjointproposal}) as the \emph{joint graphical lasso} (JGL).
We discuss specific forms of the penalty function in (\ref{loglikjointproposal}) in the next section.

\subsection{ Two useful penalty functions}
\label{penalties}
In this subsection, we introduce two particular choices of the convex penalty function $P$ in (\ref{loglikjointproposal}) that lead to useful graphical model estimates.  
In Appendix 1, we further extend these proposals to work on the scale of partial correlations.

\subsubsection{The fused graphical lasso}
The \emph{fused graphical lasso} (FGL) is the solution to the problem (\ref{loglikjointproposal}) with  the penalty 
\begin{equation}
P(\sTheta) = \lambda_1  \sum_{k=1}^K \sum_{i \neq j}
|\theta^{(k)}_{ij}| + \lambda_2  \sum_{k < k'} \sum_{i, j}
|\theta^{(k)}_{ij}-\theta^{(k')}_{ij}|,
\label{FGLpenalty}
\end{equation}
where  $\lambda_1$ and $\lambda_2$ are nonnegative tuning
parameters. This is a  
\emph{generalized fused lasso} penalty \citep{HoeflingFLSA}, and results from applying
$\ell_1$ penalties to (1) each off-diagonal element of the $K$ precision
matrices,
%as in \citet{YuanLin07}, and applying a fused lasso \citep{TSRZ2005} penalty
and (2) differences between corresponding elements of each pair of
precision matrices. Like the graphical lasso, FGL results in sparse
estimates $\bThetahat^{(1)},\ldots,\bThetahat^{(K)}$ when the tuning
parameter $\lambda_1$ is large. In addition, many elements of
$\bThetahat^{(1)},\ldots,\bThetahat^{(K)}$ will  be
 identical across classes when the tuning parameter $\lambda_2$ is large
 \citep{TSRZ2005}.
Thus FGL borrows information aggressively across classes,
encouraging not only similar network structure but also similar edge
values.

\subsubsection{The group graphical lasso}
We  define  the \emph{group graphical lasso} (GGL) to be the solution to (\ref{loglikjointproposal}) with
\begin{equation}
P(\sTheta) = \lambda_1  \sum_{k=1}^K \sum_{i \neq j}
|\theta^{(k)}_{ij}| + \lambda_2  \sum_{i \neq j} \sqrt{\sum_{k=1}^K
{\theta^{(k)}_{ij}}^2 }. \label{GGLpenalty}
\end{equation}
Again, $\lambda_1$ and $\lambda_2$ are nonnegative tuning
parameters. A lasso penalty is applied to the
elements of the precision matrices and a \emph{group lasso} penalty is  applied to
the $(i,j)$ element  across all $K$
 precision matrices \citep{grouplasso}.
This group lasso penalty encourages a similar pattern of sparsity across
all of the precision matrices -- that is, there will be a tendency
for the zeros in the $K$ estimated precision matrices to occur in
the same places. Specifically, when $\lambda_1=0$ and $\lambda_2>0$,
each $\bThetahat^{(k)}$ will have an identical pattern of non-zero
elements. On the other hand, the lasso penalty encourages further
sparsity within each $\bThetahat^{(k)}$.
% When $\lambda_2 = 0$ and
%$\lambda_1>0$, GGL is equivalent to computing a separate graphical
%lasso solution for each class.

GGL encourages a weaker form of similarity across the $K$ precision
matrices than does FGL: the latter encourages shared edge values across the $K$ matrices, whereas the former encourages only a shared pattern
of sparsity.
%. The latter encourages not only a shared
%pattern of sparsity between the $K$ precision matrices, but also
%similar values of the elements of the concentration matrices.

\section{Algorithm for the joint graphical lasso problem}
\label{alg} 

\subsection{An ADMM algorithm}

We solve the  problem 
(\ref{loglikjointproposal}) using an \emph{alternating directions method of multipliers} (ADMM) algorithm.
We refer the reader to \citet{ADMMBoyd} for a thorough exposition of ADMM algorithms as well as their convergence properties, and
to   \citet{SimonTibs12}  and \citet{Mohan12} for recent  applications of ADMM  to  related problems.
% We briefly review  GGD, and then describe
%its application to (\ref{loglikjointproposal}).

To solve the problem (\ref{loglikjointproposal}) 
%\begin{equation}
%\minimize_{\sTheta} \left\{ - \sum_{k=1}^K  n_k \left(\log \det {\bf \Theta}^{(k)} - \trace( {\bf S}^{(k)} {\bf \Theta}^{(k)})\right)  + P(\sTheta) \right\},
%%\minimize_{\sTheta}  - \sum_{k=1}^K  n_k \left(\log \det {\bf \Theta}^{(k)} - \trace( {\bf S}^{(k)} {\bf \Theta}^{(k)})\right) + P(\sTheta)  ,
%\label{loglikjointproposal2}
%\end{equation}
subject to the constraint that ${\bf \Theta}^{(k)}$ is positive definite for $k=1,\ldots,K$
using ADMM, we note that the problem can be rewritten as
\begin{equation}
\minimize_{\sTheta,\sZ} \left\{ - \sum_{k=1}^K  n_k \left(\log \det {\bf \Theta}^{(k)} - \trace( {\bf S}^{(k)} {\bf \Theta}^{(k)})\right)  + P(\sZ) \right\},
\label{loglikjointproposal2}
\end{equation}
subject to the positive-definiteness constraint as well as the constraint that ${\bf Z}^{(k)} = {\bf \Theta}^{(k)}$ for $k=1,\ldots,K$, where $\sZ = \{ {{\bf Z}^{(1)}}, \ldots,\bZK \}$.
The scaled augmented Lagrangian \citep{ADMMBoyd} for this problem is given by 
%\footnotesize
\begin{eqnarray}
L_{\rho}(\sTheta,\sZ,\sU)=  &-& \sum_{k=1}^K  n_k \left(\log \det \bThetak - \trace( \bSk \bThetak)\right)  + P(\sZ) \nonumber \\
  &+&  \frac{\rho}{2}  \sum_{k=1}^K ||\bThetak-\bZk + \bUk||_F^2,%+ \sum_{k=1}^K \trace \left( (\bUk)^T (\bThetak - \bZk)  \right).
\label{scaled.aug}
\end{eqnarray}
%\normalsize
where $\sU = \{ {{\bf U}^{(1)}}, \ldots,\bUK \}$ are dual variables.
%subject to the constraint that ${\bf \Theta}^{(k)} \succeq \delta {\bf I}$ for all $k=1,\ldots,K$.
Roughly speaking, an ADMM algorithm corresponding to (\ref{scaled.aug}) results from iterating three simple steps. At the $i$th iteration, they are as follows:
\begin{enumerate}
\item  $\{ {\bf \Theta}_{(i)} \} \leftarrow \arg\min_{\sTheta} \left\{ L_{\rho} \left(\sTheta,\{ {\bf Z}_{(i-1)} \} ,\{ {\bf U}_{(i-1)} \} \right) \right\}$.
\item $\{ {\bf Z}_{(i)} \} \leftarrow  \arg\min_{\sZ} \left\{ L_{\rho} \left(     \{ {\bf \Theta}_{(i)} \},\sZ,\{ {\bf U}_{(i-1)} \}   \right) \right\}$.
\item  $\{ {\bf U}_{(i)} \} \leftarrow \{ {\bf U}_{(i-1)} \} + (\{ {\bf \Theta}_{(i)} \} -\{ {\bf Z}_{(i)} \})$.
\end{enumerate}
We now present the ADMM algorithm in greater detail.
\begin{center}
\textbf{ADMM algorithm for solving the joint graphical lasso problem}
\begin{enumerate}
\item Initialize the variables: $\bThetak = {\bf I}$, $\bUk ={\bf 0}$, $\bZk ={\bf 0}$ for $k=1,\ldots,K$.
\item Select a scalar $\rho>0$.
%\item Initialize $s$, the step size, and choose some $\eta < 1$.
\item For $i=1,2,3,\ldots$ until convergence:
\begin{enumerate}
\item  For $k=1,\ldots,K$, update ${\bf \Theta}_{(i)}^{(k)}$ as the minimizer (with respect to ${\bf \Theta}^{(k)}$) of
$$-n_k \left(\log \det {\bf \Theta}^{(k)} - \trace( {\bf S}^{(k)} {\bf \Theta}^{(k)})\right)  
 + \frac{\rho}{2} ||\bThetak - {\bf Z}_{(i-1)}^{(k)} + {\bf U}_{(i-1)}^{(k)}  ||_F^2.$$ 
 Letting ${\bf V}{\bf D}{\bf V}^T$ denote the eigendecomposition of $\bSk-{\rho} {\bf Z}_{(i-1)}^{(k)} / n_k+{\rho} {\bf U}_{(i-1)}^{(k)}   / n_k$, the solution is given
 \citep{WittenTibsScout09}  by ${\bf V} \tilde{\bf D} {\bf V}^T$, where $\tilde{\bf D}$ is the diagonal matrix with $j$th diagonal element 
$$\frac{n_k}{2 \rho} \left(-D_{jj} + \sqrt{D_{jj}^2 + 4 \rho/n_k}\right).$$ %\citep{WittenTibsScout09}.
%
%FINISH!!
\item Update $\{ {\bf Z}_{(i)} \}$ as the minimizer (with respect to $\sZ$) of
\small
\begin{equation}
    \frac{\rho}{2}  \sum_{k=1}^K ||\bZk - ({\bf \Theta}_{(i)}^{(k)} + {\bf U}_{(i-1)}^{(k)}  )||_F^2 +  P(\sZ). 
  \label{goodlord}
\end{equation}
\normalsize
%Details of this minimization will depend on the  form of the convex penalty function $P$, and will be considered in the following sections.
%subject to the constraint that ${\bf \Theta}^{(k)} \succeq {\bf 0}$ for $k=1,\ldots,K$.
\item 
For $k=1,\ldots,K$, update ${\bf U}_{(i)}^{(k)}$ as ${\bf U}_{(i-1)}^{(k)}  + ({\bf \Theta}_{(i)}^{(k)} - {\bf Z}_{(i)}^{(k)})$.
% - \sum_{k=1}^K  n_k \left(\log \det {\bf \Theta}^{(k)}_{(i)} - \trace( {\bf S}^{(k)} {\bf \Theta}^{(k)}_{(i)})) \right)  \leq   \nonumber \\
% - \sum_{k=1}^K  n_k \left(\log \det {\bf \Theta}^{(k)}_{(i-1)} - \trace( {\bf S}^{(k)} {\bf \Theta}^{(k)}_{(i-1)}) \right) +  \nonumber \\
% \langle(\bTheta^{(k)}_{(i)} - \bTheta^{(k)}_{(i-1)}),( (\bTheta^{(k)}_{(i-1)})^{-1} - {\bf S}^{(k)} )\rangle  +
%\frac{1}{2s} || \bTheta^{(k)}_{(i)} - \bTheta^{(k)}_{(i-1)} ||^2   \nonumber
%\label{backtrack}
%\end{eqnarray}
%\normalsize
%\textcolor{red}{I haven't corrected what I believe is a sign error in the previous equation.}
%If the check fails, discard $\{ {\bf \Theta}_{(i)} \}$, multiply the step size $s$ by $\eta$, and return to step (c)(i).

%\item Check that each $\bTheta^{(k)}_{(i)}$ is positive definite, and check that the objective  (\ref{loglikjointproposal2}) evaluated at $\{ {\bf \Theta}_{(i)} \}$ is smaller than the objective function evaluated at $\{ {\bf \Theta}_{(i-1)} \}$.
%If either of these checks fails, halve the step size $s$ and return to step 3(a).

%where $s_i = s_{i-1}/\eta^j$, with $\eta>1$. Here  $j$ is the smallest nonnegative integer such that
%${\bf \Theta}^{(k)}_{(i)}$ is positive definite for $k=1,\ldots,K$, and (\ref{loglikjointproposal2}) evaluated at
%$\{ {\bf \Theta}_{(i)} \}$ is less than (\ref{loglikjointproposal2}) evaluated at $\{ {\bf \Theta}_{(i-1)} \}$.
%\item Update $i \leftarrow i+1$.
\end{enumerate}
\end{enumerate}
\end{center}
The final $\hat\bTheta^{(1)}, \ldots, \hat\bTheta^{(K)}$ that result from this algorithm are the JGL estimates of $\siginv_1,\ldots,\siginv_K$.
This algorithm is guaranteed to converge to the global optimum \citep{ADMMBoyd}. We note that the positive-definiteness constraint on the estimated precision matrices is naturally 
enforced by the update in Step (c)(i).

Details of the minimization of (\ref{goodlord})  will depend on the  form of the convex penalty function $P$.
We note that  the task of minimizing (\ref{goodlord}) can be re-written as
%\begin{equation}
%\minimize_{\sTheta} \left\{    \frac{1}{2s} \sum_{k=1}^K ||{\bf
%\Theta}^{(k)} - {\bf A}^{(k)} ||_F^2  + P(\sTheta)  \right\},
%\label{simpler}
%\end{equation}
\begin{equation}
 \minimize_{\sZ} \left\{   \frac{\rho}{2}  \sum_{k=1}^K ||\bZk - {\bf A}^{(k)}||_F^2 +  P(\sZ) \right\},
  \label{simpler}
\end{equation}
%subject to the constraint that ${\bf \Theta}^{(k)} \succeq \delta {\bf I}$ for $k=1,\ldots,K$, where
where
\begin{equation}
{\bf A}^{(k)} = {\bf \Theta}_{(i)}^{(k)} + {\bf Y}_{(i-1)}^{(k)}.
%{\bf \Theta}^{(k)}_{(i-1)} +  n_k s \left( ({\bf
%\Theta}^{(k)}_{(i-1)})^{-1} - {\bf S}^{(k)}   \right). 
\label{A}
\end{equation}
We will see in Section \ref{FGLGGLsol}
 that for the FGL and GGL penalties, solving (\ref{simpler}) is a simple task, regardless of the value of $K$.
%In our experience, the algorithm tends to converge in fewer than 100 iterations.

The algorithm given above involves computing the eigen decomposition of  a $p \times p$ matrix, which can be computationally demanding when $p$ is large. However, in Section \ref{screen}, we will present two theorems that reveal that when the values of the tuning parameters $\lambda_1$ and $\lambda_2$  are large, one can obtain
the \emph{exact} solution to the JGL optimization problem without ever computing the eigen decomposition of  a $p \times p$ matrix.
Therefore, solving the JGL problem is fast even when $p$ is quite large. In Section \ref{realstudy}, we will see that one can perform FGL with $K=2$ classes  and almost 18,000 features in under 2 minutes.

\subsection{Solving (\ref{simpler}) for the joint graphical lasso}
\label{FGLGGLsol}
% We saw in the previous section that we can solve
%the JGL problem by repeatedly solving (\ref{goodlord}), or more
%simply (\ref{simpler}).
%, subject to a constraint that ${\bf \Theta}^{(k)} \succeq \delta {\bf I}$ for $k=1,\ldots,K$. 
We now consider the problem of solving
(\ref{simpler}) if $P$ is a generalized fused lasso or
group lasso penalty.
% In Sections~\ref{fgl.ggd.update} and \ref{ggl.ggd.update} we present a solution to (\ref{goodlord})  without the constraint on the smallest eigenvalue
%of ${\bf \Theta}^{(k)}$, and in Section~\ref{smallesteig} we will present a solution for solving (\ref{simpler}) subject to this constraint.

\subsubsection{Solving (\ref{simpler}) for FGL}
\label{fgl.ggd.update}
If $P$ is the  penalty given in (\ref{FGLpenalty}), then %fff
(\ref{simpler}) takes the form
\begin{equation}
\minimize_{\sZ} \left\{    \frac{\rho}{2} \sum_{k=1}^K ||{\bf Z}^{(k)} - {\bf A}^{(k)} ||_F^2  + \lambda_1  \sum_{k=1}^K \sum_{i \neq j} |Z^{(k)}_{ij}|
+ \lambda_2  \sum_{k < k'} \sum_{i,j} |Z^{(k)}_{ij}-Z^{(k')}_{ij}|  \right\}.
\label{simpler.fgl}
\end{equation}
Now (\ref{simpler.fgl}) is completely separable with respect to each pair of matrix elements $(i,j)$: that is, one can simply solve, for each $(i,j)$,
\begin{equation}
\minimize_{Z_{ij}^{(1)}, \ldots, Z_{ij}^{(K)}} \left\{    \frac{\rho}{2} \sum_{k=1}^K (Z_{ij}^{(k)} - {A}_{ij}^{(k)} )^2  + \lambda_1  1_{i \neq j} \sum_{k=1}^K |Z^{(k)}_{ij}|
+ \lambda_2  \sum_{k < k'}  |Z^{(k)}_{ij}-Z^{(k')}_{ij}|  \right\}.
\label{evensimpler.fgl}
\end{equation}
This is a special case of  the \emph{fused lasso signal
approximator} \citep{HoeflingFLSA}  in which there is a fusion
between each pair of variables. A very efficient algorithm  for this
special case, which can be performed in $O(K \log K)$ operations,
is available   \citep{Hocking2011,FLSA_closed_form_sol,FLSA_closed_form_sol_RyanTibs}.

In fact, when $K=2$, (\ref{evensimpler.fgl}) has a very simple closed form solution. When
 $\lambda_1=0$, it is easy to verify that the solution to (\ref{evensimpler.fgl}) takes the form
\begin{equation}
(\hat{Z}_{ij}^{(1)},\hat{Z}_{ij}^{(2)}) = \left\{ \begin{array}{lll}
 (A_{ij}^{(1)}- \lambda_2/\rho,A_{ij}^{(2)}+ \lambda_2/\rho) & \mbox{if $A_{ij}^{(1)} > A_{ij}^{(2)} + 2  \lambda_2/\rho$} \\
 (A_{ij}^{(1)}+ \lambda_2/\rho,A_{ij}^{(2)}- \lambda_2/\rho) & \mbox{if $A_{ij}^{(2)} > A_{ij}^{(1)} + 2  \lambda_2/\rho$} \\
 (\frac{A_{ij}^{(1)}+A_{ij}^{(2)}}{2},\frac{A_{ij}^{(1)}+A_{ij}^{(2)}}{2}) & \mbox{if $|A_{ij}^{(1)} - A_{ij}^{(2)}| \leq 2  \lambda_2/\rho$} \\
\end{array}.
\right.
\label{prethresholdsol}
\end{equation}
And when $\lambda_1>0$, the solution to (\ref{evensimpler.fgl}) %when $\lambda_1>0$
can be obtained through soft-thresholding (\ref{prethresholdsol}) by
$\lambda_1/\rho$ \cite[see][]{PathwiseCoord}. Here the
soft-thresholding operator is defined  as $S(x,c)=\sgn(x)(|x|-c)_+$,
where $a_+ = \max(a,0)$.

\subsubsection{Solving (\ref{simpler}) for GGL}
\label{ggl.ggd.update}
If  $P$ is the group lasso penalty (\ref{GGLpenalty}), then
(\ref{simpler}) takes the form
\begin{equation}
\minimize_{\sZ} \left\{    \frac{\rho}{2} \sum_{k=1}^K ||{\bf Z}^{(k)} - {\bf A}^{(k)} ||_F^2  +
\lambda_1  \sum_{k=1}^K \sum_{i \neq j} |Z^{(k)}_{ij}| + \lambda_2  \sum_{i \neq j} \sqrt{\sum_{k} {Z^{(k)}_{ij}}^2 } \right\}.
\label{simpler.ggl}
\end{equation}
First, for all  $i=1,\ldots,p$ and $k=1,\ldots,K$,  it is easy to see that the solution to (\ref{simpler.ggl}) has $\hat{Z}^{(k)}_{ii} =
A^{(k)}_{ii}$. And one can show that the off-diagonal elements take the form  \citep{SparseGroupLasso}
\begin{equation}
    \hat{Z}^{(k)}_{ij} = S(A^{(k)}_{ij}, \lambda_1/\rho )  \left(1 - \frac{ \lambda_2}{\rho \sqrt{ \sum_{k=1}^K S(A^{(k)}_{ij}, \lambda_1/\rho)^2 } } \right)_+,
\label{diaggglsoln}\end{equation} where $S$ denotes the
soft-thresholding operator.

%\subsubsection{Solving (\ref{simpler}) subject to a constraint}
%\label{smallesteig} 
%Now, recall that the GGD algorithm involves solving (\ref{simpler})   subject to a constraint that ${\bf \Theta}^{(k)} \succeq \delta {\bf I}$ for $k=1,\ldots,K$.
%It will often be the case that the solution to the unconstrained optimization problem satisfies this constraint, in which case the approaches to solving 
%(\ref{simpler}) outlined in Sections~\ref{fgl.ggd.update} and \ref{ggl.ggd.update} suffice to solve the constrained problem.
%However, if the unconstrained optimization problem yields a solution that violates the constraint, then an \emph{alternating direction method of multipliers} (ADMM)
%\citep{ADMMBoyd} algorithm can be used to solve the constrained problem (\ref{simpler}). %This is detailed in Appendix 6.
%Briefly, (\ref{simpler}) subject to the constraint can be re-written as 
%\begin{eqnarray}
%\minimize_{\sTheta, \sZ} &&\left\{    \frac{1}{2s} \sum_{k=1}^K ||{\bf
%\Theta}^{(k)} - {\bf A}^{(k)} ||_F^2  + P(\sZ)  \right\} \nonumber \\
%&&\mbox{ subject to } {\bf \Theta}^{(k)} \succeq \delta {\bf I}, {\bf Z}^{(k)}={\bf \Theta}^{(k)} \mbox{ for all } k=1,\ldots,K,% \mbox{ and } \sZ=\sTheta.
%\label{admm}
%\end{eqnarray}
%which can be easily solved via ADMM.
%% because solving (\ref{admm}) with respect to $\sTheta$ with $\sZ$ held fixed is straightforward, and we have already presented the solution 
%%to (\ref{admm}) with $\sTheta$ held fixed. 

%We note that a  closely-related approach is taken in \citet{BienTibs10}.
\section{Faster computations for FGL and GGL}
\label{screen}

We now present two theorems that lead to substantial computational improvements to the JGL algorithm presented in Section
\ref{alg}.
Using these theorems, one can inspect the empirical covariance matrices ${\bf S}^{(1)}, \ldots, {\bf S}^{(K)}$ in order to determine whether the solution
to the JGL
optimization problem is block diagonal
 after some permutation of the features.
Then one can simply perform the JGL algorithm on the features within each block separately, in order to obtain \emph{exactly} the same solution that would have been obtained by
applying the algorithm to all $p$ features. This leads to huge speed improvements since it obviates the need to ever compute the eigen decomposition of a $p \times p$ matrix.
Our results mirror recent
improvements in
algorithms for solving the graphical lasso problem  \citep{WittenFriedman11,MazumderHastie2012}.

 For instance, suppose that for a given choice of $\lambda_1$ and $\lambda_2$,  we determine that the estimated inverse covariance matrices $\hat{\bf \Theta}^{(1)}, \ldots, \hat{\bf \Theta}^{(K)}$
  are block diagonal, each with the same  $R$  blocks, the $r$th of which contains $p_r$ features,  $\sum_{r=1}^R p_r = p$. Then  in each iteration of the JGL algorithm,
rather than having to compute the eigen decomposition of  $K$ $p \times p$ matrices, we need only compute eigen decompositions of  matrices of dimension $p_1 \times p_1, \ldots, p_R \times p_R$. This leads to a potentially massive reduction in computational complexity from $O(p^3)$ to $\sum_{r=1}^R O(p_r^3)$.

We begin with a very simple lemma for which the proof follows by inspection of (\ref{loglikjointproposal}).
The lemma can be extended by induction to any number of blocks.

\begin{lemma}
\label{screeninglemma}
Suppose that the solution to the FGL or GGL optimization problem is block diagonal with known blocks. That is,
the features can be reordered in such a way that each estimated inverse covariance matrix takes the form
\begin{equation}
\bThetahat^{(k)} = \left( \begin{array}{cc} \bThetahat^{(k)}_{1} & 0 \\ 0 & \bThetahat^{(k)}_{2} \end{array} \right)
\label{soln}
\end{equation}
where each of  $\bThetahat^{(1)}_1, \ldots, \bThetahat^{(K)}_1$ has the same dimension.
Then, $\bThetahat^{(1)}_1, \ldots, \bThetahat^{(K)}_1$ and  $\bThetahat^{(1)}_2, \ldots, \bThetahat^{(K)}_2$
can be obtained by solving the FGL or GGL optimization problem on just the corresponding set of features.
\end{lemma}

We now present the key results.
Theorems \ref{thm2}  and \ref{thm3} outline necessary and sufficient conditions for the presence of block diagonal structure in the FGL and GGL optimization problems,
and are proven in Appendix 2.

\begin{theorem}
\label{thm2}
Consider the FGL optimization problem with $K=2$ classes. Let $C_1$ and $C_2$ be a non-overlapping partition of the $p$ variables such that $C_1 \cap C_2 = \emptyset$, $C_1 \cup C_2 = \{ 1, \ldots, p \}$.
The following conditions are necessary and sufficient for the variables in $C_1$ to be completely disconnected from those in $C_2$ in each of the resulting  network estimates:
\begin{enumerate}
\item $|n_1 S^{(1)}_{ij}| \leq \lambda_1 + \lambda_2$ for all $i \in C_1$ and $j \in C_2$,
\item $|n_2 S^{(2)}_{ij}| \leq \lambda_1 + \lambda_2$ for all $i \in C_1$ and $j \in C_2$, and
\item$|n_1 S^{(1)}_{ij}+ n_2 S^{(2)}_{ij}| \leq 2\lambda_1$ for all $i \in C_1$ and $j \in C_2$.
\end{enumerate}
Furthermore, if $K>2$, then
\begin{equation}
|n_k S^{(k)}_{ij}| \leq \lambda_1\mathrm{\;\; for\; all \;\;} i \in C_1, \; j \in C_2, \; k=1,\ldots,K
\end{equation}
is a sufficient condition for the variables in $C_1$ to be completely disconnected from those in $C_2$.
\end{theorem}

\begin{theorem}
\label{thm3}
Consider the GGL optimization problem with $K \geq 2$ classes.
 Let $C_1$ and $C_2$ be a non-overlapping partition of the $p$ variables, such that 
$C_1 \cap C_2 = \emptyset$, $C_1 \cup C_2 = \{ 1, \ldots, p \}$.
 The following condition is necessary and sufficient for the variables in $C_1$ to be completely disconnected from those in $C_2$ in each of the resulting network estimates:
\begin{equation}
\sum_{k=1}^K (|n_kS^{(k)}_{ij}|-\lambda_1)_+^2 \leq \lambda_2^2 \mathrm{\;\; for\; all \;\;} i \in C_1, \; j \in C_2 \;.
\label{thm3condition}
\end{equation}
\end{theorem}
Theorems \ref{thm2} and \ref{thm3} allow us to quickly check if,
given a partition of the features $C_1$ and $C_2$, the solution to
the JGL optimization problem is block diagonal with one block
corresponding to features in $C_1$ and one block corresponding to
features in $C_2$. In practice, for any given
$(\lambda_1, \lambda_2)$,
%however, we do not know in advance
%the partition $C_1$ and $C_2$. Instead,
we can quickly perform the following two-step procedure to identify
any block structure in the FGL or GGL solution:
\begin{enumerate}
\item Create $\bf M$, a $p \times p$ matrix with $M_{ii}=1$ for $i=1,\ldots,p$.
For $i \neq j$, let $M_{ij}=0$ if  the conditions specified in Theorem \ref{thm2} are met for that pair of variables and the FGL penalty is used, or if the condition of Theorem \ref{thm3} is met for that pair of variables and the GGL penalty is used.
Otherwise, set $M_{ij}=1$.
%\item Create $\bf M$, a $p \times p$ matrix with $M_{ii}=1$ for $i=1,\ldots,p$. For $i \neq j$, let $M_{ij}=0$ if
% $\sum_{k=1}^K (|n_kS^{(k)}_{ij}|-\lambda_1)_+^2 \leq \lambda_2^2$ for all $k=1,\ldots,K$ (if performing GGL) or if $|n_k S^{(k)}_{ij}| \leq \lambda_1$ for all $k =1, \ldots, K$  (if performing FGL).  Otherwise, $M_{ij}=1$.
 \item Identify the connected components of the undirected graph whose adjacency matrix is given by $\bf M$. Note that this can be performed in $O(|M|)$ operations, where $|M|$ is the number of nonzero elements in  $\bf M$ \citep{Tarjan72}.
 \end{enumerate}
 Theorems \ref{thm2} and \ref{thm3} guarantee that the connected components identified in Step (b) correspond to distinct blocks in the FGL or GGL solutions.
 Therefore, one can quickly obtain these solutions by solving the
 FGL or GGL optimization problems on the submatrices of these $K$ $p \times p$ empirical 
 covariance matrices that correspond to that block diagonal structure. Consequently,
%JGL can be efficiently applied to very large data sets, including  even data sets whose high dimension prohibits storage of a $p \times p$ matrix in memory!
%In other words,
we  can  obtain the \emph{exact} solution to the JGL optimization
problem on extremely high-dimensional data sets that would otherwise be computationally
intractable. For instance, in Section \ref{realstudy} we performed
FGL on a gene expression data set with almost $18,000$ features in
under two minutes.

As pointed out by a reviewer, Theorems \ref{thm2} and \ref{thm3}
lead to speed improvements only if the tuning parameters $\lambda_1$
and $\lambda_2$ are sufficiently large. We argue that this will in fact be
the case in most practical applications of JGL. When network
estimation is performed for the sake of data exploration and when
$p$ is large, only a very sparse network estimate will be useful; otherwise,
interpretation of the estimate will be impossible.
  Even when data exploration is not the end goal of the analysis, large values of $\lambda_1$ and $\lambda_2$ will generally be used, since most data sets
cannot reasonably support estimation of $Kp(p+1)/2$ nonzero parameters when $n \ll p$.

\section{Relationship to previous proposals}
\label{related} Several past proposals have been made to jointly
estimate graphical models on the basis of observations drawn from
distinct conditions.
Some proposals have used time-series data
to define time-varying networks in the context of continuous or binary data
 \citep{Zhou2008,KELLER,AhmedTESLA,Kolar2009,Song2009,KolarXing2010}.
 \citet{Guo2011} instead describe a likelihood-based method for estimating precision matrices across multiple related classes simultaneously.
They employ a hierarchical penalty that forces similar patterns of sparsity across classes, an approach that is similar in spirit to GGL.

Our FGL and GGL proposals have a number of advantages over these existing approaches.
Methods for estimating time-varying networks cannot be easily extended to the setting where the classes lack a natural ordering.
\citet{Guo2011}'s proposal is a closer precursor to our method, and can in fact be stated as an instance of the problem (\ref{loglikjointproposal}) with a \emph{hierarchical group lasso penalty} % penalty function
\begin{equation}
P(\sTheta) = \lambda  \sum_{i \neq j} \sqrt{\sum_{k} {|\theta^{(k)}_{ij}}| }
\label{guo}
\end{equation}
 that encourages a shared pattern of sparsity across the $K$ classes.
But the approach of \citet{Guo2011} has a number of  disadvantages relative to FGL and GGL.
(1) The penalty (\ref{guo}) is not convex, so convergence to the wrong local maximum is possible. (2) Because (\ref{guo}) is not convex,
 it is not possible to  achieve
 the speed improvements described in Section \ref{screen}. Consequently, the \citet{Guo2011} proposal is quite slow relative to our approach, as seen in Figures \ref{Sixplots}(e), \ref{Sixplots.K2}(e), and \ref{Sixplots.onebignet}(e), and
 essentially cannot be applied to very high-dimensional data sets.
(3) Unlike FGL and GGL, it uses just one tuning parameter, and is unable to control separately the sparsity level and the extent of network similarity.
(4) In cases where we expect edge values as well as network structure to be similar between classes, FGL is much better suited than
 GGL and \citet{Guo2011}'s proposal, both of which encourage shared patterns of sparsity but ignore the signs and values of the nonzero edges.

 \citet{Guo2011}'s proposal is included in the simulation study in Section \ref{simstudy}.

\section{Tuning parameter selection}
\label{tuning_parameter_selection}
One can  select tuning parameters for JGL  using an
approximation of the Akaike Information Criterion (AIC),
\begin{equation}
AIC(\lambda_1,\lambda_2) = \sum_{k=1}^K \left[n_k \trace( {\bf S}^{(k)} \hat{\bTheta}^{(k)}_{\lambda_1,\lambda_2})
 - n_k \log \det \hat{\bTheta}^{(k)}_{\lambda_1,\lambda_2} + 2 E_k \right],
 \label{AIC}
\end{equation}
where $\{\hat{\bTheta}^{(k)}_{\lambda_1,\lambda_2}\}$ is the set of
estimated inverse covariance matrices based on tuning parameters
$\lambda_1$ and $\lambda_2$,
 and $E_k$ is the number of non-zero elements in $\hat{\bTheta}^{(k)}_{\lambda_1,\lambda_2}$.
A grid search can then be performed to select $\lambda_1$ and $\lambda_2$ minimizing the $AIC(\lambda_1,\lambda_2)$
   score.
%Figure \ref{Sixplots} indicates
The simulation study in Section 7 suggests that this criterion tends
to select models whose Kullback-Leibler (dKL) divergence from the
true model is low. When the number of variables $p$ is very large,
computing $AIC(\lambda_1,\lambda_2)$ over a range of values for
$\lambda_1$ and $\lambda_2$ may prove computationally onerous. If
this is the case, we suggest a dense search over
$\lambda_1$ followed by a quick search over $\lambda_2$.
% the ``sparsity"
%tuning parameter  followed by a quick search over the ``similarity"
%tuning parameter.

It is worth noting that in most cases, network estimation is performed as a part of exploratory data analysis and hypothesis generation.
For these purposes, approaches such as AIC, BIC, and cross-validation may tend to choose models too large to be useful.
In this setting, model selection should be guided by practical considerations, such as network  interpretability, stability, and
 the desire for an edge set with a  low false discovery rate  \citep{StabilitySelection,BINCO}.

\section{Simulation study}
\label{simstudy}

We  compare the performances of FGL and GGL
 to two existing methods, graphical lasso and \citet{Guo2011}'s proposal, in Section \ref{K3sims}.
 When applying the graphical lasso, networks are fitted for
each class separately.
We investigate the effects of $n$ and $p$ on FGL and GGL's
performances
 in Section \ref{three}. Additional simulation results are presented in Appendix 3.

The effects of the FGL and GGL penalties vary with
the sample size.  For ease of presentation of the simulation study results, we multiply the reported tuning parameters $\lambda_1$ and $\lambda_2$ by the sample size of each class before performing JGL.
%(Since $\omega_2$ represents the proportion of the sparsity penalty in GGL that is due to the group lasso penalty, it does not require scaling.)

To ease interpretation, we reparametrize the GGL penalties in our simulation study. 
The motivation is to
summarize the regularization for ``sparsity" and for ``similarity"
separately. In FGL, this is nicely achieved by just using
$\lambda_1$ and $\lambda_2$, as the former drives network sparsity and the
latter drives network similarity. In contrast, in GGL, both tuning
parameters contribute to sparsity: $\lambda_1$ drives individual
network edges to zero whereas $\lambda_2$ simultaneously drives
network edges to zero across all $K$ network estimates. 
We reparameterize our simulation results for GGL in terms of
$\omega_1 = \lambda_1 + \frac{1}{\sqrt{2}}\lambda_2$ and
$\omega_2 = \frac{1}{\sqrt{2}}\lambda_2 / (\lambda_1 +
\frac{1}{\sqrt{2}} \lambda_2)$, which we found to reasonably reflect the levels of 
``sparsity" and ``similarity" regularization, respectively. 

\subsection{Performance as a function of tuning parameters}
\label{K3sims}

\subsubsection{Simulation set-up}
\label{setup} In this simulation, we consider a three-class problem.
We first generate three networks with $p=500$ features belonging to
ten equally sized unconnected subnetworks, each with a power law
degree distribution. Power law degree distributions are thought to
mimic the structure of biological networks \citep{Chen2004} and are
generally harder to estimate than simpler structures \citep{Space}. Of the ten
subnetworks, eight have the same structure and edge values in all
three classes, one is identical between the first two classes and
missing in the third (i.e. the corresponding features are singletons
in the third network), and one is present in only the first class.
The topology of the networks generated is shown in Figure
\ref{net_used_in_sims} in Appendix 4.

Given a network structure, we generate a covariance matrix for the first class as follows \citep{Space}.
We create a $p \times p$ matrix with ones on the diagonal, zeroes on elements not corresponding to network edges,
 and values from a uniform distribution with support on $\{[-.4,-.1]\cup[.1,.4]\}$ on elements corresponding to edges.
To ensure positive definiteness, we divide each off-diagonal element by 1.5 times the sum of the absolute values of off-diagonal elements in its row.
Finally, we average the matrix with its transpose, achieving a symmetric, positive-definite matrix $\bf A$.
We then create the $(i,j)$ element of ${\bf \Sigma}_1$ as 
\[
 d_{ij}  ({\bf A}^{-1})_{ij} / \sqrt{({\bf A}^{-1})_{ii}  ({\bf A}^{-1})_{jj} },
\]
where $d_{ij} = 0.6$ if $i\neq j$ and $d_{ij}=1$ if $i=j$. We create
${\bf \Sigma}_2$ equal to ${\bf \Sigma}_1$, then reset one
of its ten subnetwork blocks to the identity. We create ${\bf
\Sigma}_3$ equal to ${\bf \Sigma}_2$, and reset an
additional subnetwork block to the identity. Finally, for each class
we generate independent, identically distributed samples from a
$N({\bf 0},{\bf \Sigma}_k)$ distribution.
%The values of $n_1=n_2$ and $p$ vary between simulations.

We present two additional simulations studies involving  two-class datasets in Appendix 3.
The first additional simulation uses the same network structure described above, and the second uses a single power law network with no block structure.

\subsubsection{Simulation results}
\label{one} Our first set of simulations illustrates the effect of
varying tuning parameters on the performances of FGL and GGL. We
generated 100 three-class data sets with $p=500$ features and $n=150$ observations
per class, as described in Section \ref{setup}.
Class 1's network had 490 edges, class 2's network is missing 49 of those edges, and  class 3's network is missing an
additional 49 edges.
Figure \ref{Sixplots} shows the results, averaged over the 100 data sets.
In each plot, the lines for FGL and for GGL indicate the results obtained with a single value of the similarity tuning parameters $\lambda_2$ and $\omega_2$.
The graphical lasso and the proposal of \citet{Guo2011} are included in the comparisons.

Figure \ref{Sixplots}(a) displays the number of true edges selected against
the number of false edges selected.
%To account for edges for which the methods' algorithms declare convergence at values
% just short of the true optimum of zero,
%we say the edge represented by $\theta_{ij}^{(k)}$ is selected if
%$\hat\theta_{ij}^{(k)} \neq 0$.
As the sparsity tuning parameters
$\lambda_1$ and $\omega_1$ decrease, the number of edges selected increases.
At many values of the similarity tuning
parameter $\lambda_2$, FGL dominates the other methods. At some
choices of the similarity tuning parameter $\omega_2$, GGL performs
as well as \citet{Guo2011}. FGL, GGL, and
\citet{Guo2011}'s proposal dominate the graphical lasso.

Figure \ref{Sixplots}(b) displays the sum of  squared errors (SSE)
between estimated edge values and true edge values: $\sum_{k=1}^K
\sum_{i\neq j}(\hat{\theta}^{(k)}_{ij} - (\siginv_{k})_{ij})^2$.
Unlike the proposal of \citet{Guo2011}, FGL, GGL, and the graphical
lasso tend to overshrink edge values towards zero due to the use of convex penalty functions.
%FGL and GGL share the graphical lasso's property of biasing edge values towards zero, which
%the proposal of \citet{Guo2011} does not do as strongly.
Thus, while FGL and GGL attain SSE values that are as low as those
of \citet{Guo2011}, they do so when estimating much larger networks.
When simultaneous  edge selection and estimation
are desired,  it may be useful to run FGL or GGL once and then to re-run them 
with smaller penalties on the selected edges, as  in
\citet{Relaxo07}.

Figure \ref{Sixplots}(c) evaluates each method's success in detecting \emph{differential edges}, or  edges that differ between classes.
For FGL, the number of differential edges is computed as the number of pairs $k<k', i<j$ such that  $\hat\theta_{ij}^{(k)} \neq \hat\theta_{ij}^{(k')}$.
Since GGL, the proposal of \citet{Guo2011}, and the graphical lasso cannot yield edges that are exactly identical across classes, for those approaches the number of differential edges is computed as the number of pairs $k<k', i<j$ such that  $|\hat\theta_{ij}^{(k)}-\hat\theta_{ij}^{(k')}|
> 10^{-2}$.
The number of true positive differential edges is plotted against the number of false positive differential edges.
%where a \emph{differential edge}
% is defined as an edge for which $|\hat\theta_{ij}^{(k)}-\hat\theta_{ij}^{(k')}| > 10^{-2}$ for some $k<k'$.
%(Since GGL, the proposal of \citet{Guo2011} and the graphical lasso cannot fuse edges exactly,
% we declare edge value equality for small differences to allow them to compete.)
Note that by controlling the total number of non-zero edges, the sparsity tuning parameters
$\lambda_1$ and $\omega_1$ have a large effect on the number of
edges that are estimated to differ between the two networks. FGL yields fewer false positives than the competing
methods, since it shrinks
between-class differences in edge values to zero. Since neither GGL nor \citet{Guo2011}'s method are designed to shrink
edge values towards each other,
by this measure neither method outperforms even the graphical lasso.
%Note that the models that perform well at identifying differential edge values are much smaller than those selected by AIC; consequently the AIC-selected models fall outside the borders of Figure \ref{Sixplots}(c).

Figure \ref{Sixplots}(d) displays the sum of the dKL's of the estimated distributions from the true distributions, as a function of the $\ell_1$ norm of the off-diagonal elements of the estimated precision matrices, i.e. $\sum_k \sum_{i \neq j} |\hat\theta^{(k)}_{ij}|$.
The dKL from the multivariate normal model with inverse covariance estimates $\bThetahat^{(1)}, \ldots, \bThetahat^{(K)}$ to the multivariate normal model with the true precision matrices $\siginv_1,\ldots,\siginv_K$ is
%\frac{1}{2} \sum_{k=1}^K \left( \log \left( \frac{ \det((\bThetahat^{(k)})^{-1})}{\det({\bf \Sigma}^{(k)})} \right)
%+ \trace(\bThetahat^{(k)} {\bf \Sigma}^{(k)} )    \right).
\[
\frac{1}{2} \sum_{k=1}^K \left( -\log \det(\bThetahat^{(k)}{\bf \Sigma}_k)
+ \trace(\bThetahat^{(k)} {\bf \Sigma}_k )  \right).
\]
At most values of $\lambda_2$, FGL attains a lower dKL than the other methods, followed by \citet{Guo2011}'s method,
 then by GGL.
%With the best choice of $\omega_2$, GGL slightly outperforms \citet{Guo2011}'s method.
The graphical lasso has the worst performance, since it estimates
each network separately. 
%Note that in this setting, model complexity (measured by $\ell_1$ norm, and controlled by the sparsity tuning parameter) plays as much of a role in the dKL obtained as the choice of the similarity parameter and the choice of network estimation approach {(\color{red} this sentence is too vague. Shall we just remove it?)}.

Figure \ref{Sixplots}(e) compares the methods' running times.
Computation time (in seconds) is plotted against the total number of non-zero edges estimated.  The graphical lasso is fastest, but FGL and GGL are much faster than the proposal of
\citet{Guo2011}, due to the results from Section \ref{screen}. 
%It takes less than a minute to estimate a sparse trio of networks with 500 variables using FGL or GGL. % note that Guo uses glasso.
%The identities of the models chosen by AIC are not relevant to this plot and are not displayed.
Timing comparisons were performed on an Intel Xeon x5680 3.3 GHz processor.
It is worth mentioning that the FGL algorithm is much faster in problems with only two classes,
since in that case there is a closed-form solution to the generalized fused lasso problem (Section \ref{FGLGGLsol}).  This can be seen in Figures \ref{Sixplots.K2}(e) and \ref{Sixplots.onebignet}(e) in Appendix 3.

We examined the FGL and GGL models with tuning parameters selected as described in Section \ref{tuning_parameter_selection}.  
For FGL, AIC selected the tuning parameters $\lambda_1 = 0.175$, $\lambda_2 = 0.025$.
Over the 100 replicate datasets, the  FGL models with these tuning parameters averaged a dKL of 774, 
 884 true positive edges, 2406 false positive edges, 77  true positive differential edges, and 4977 false positive differential edges.  
In GGL, AIC selected the tuning parameters $\omega_1 = 0.225$, $\omega_2 = 1$.
Over the 100 replicate datasets, the  GGL models with these tuning parameters averaged a dKL of 776, 
 898 true positive edges, 736 false positive edges, 53 true positive differential edges, and  1456 false positive differential edges.

\begin{figure}[htp]
\centering
\includegraphics[width=0.3\linewidth]{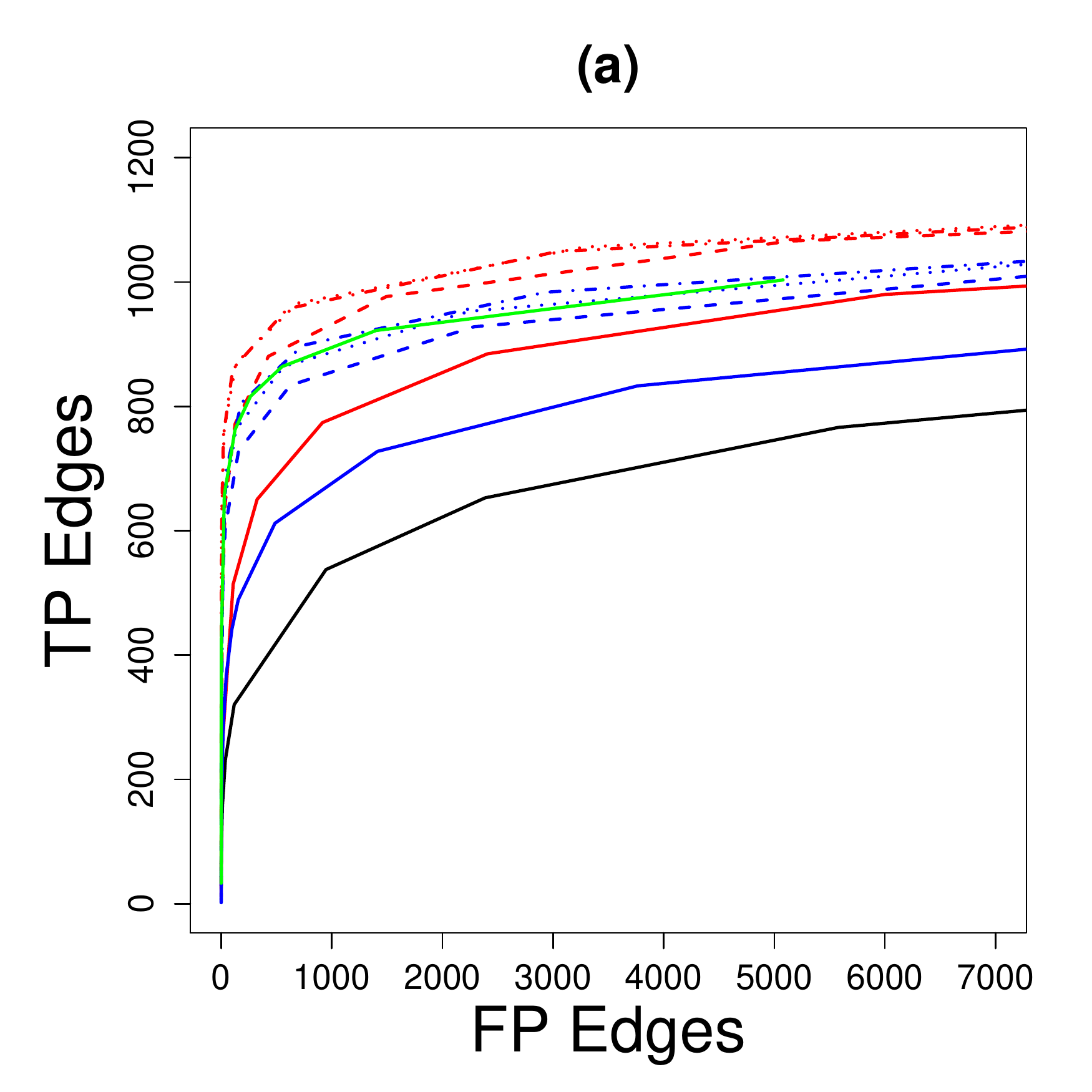}
\includegraphics[width=0.3\linewidth]{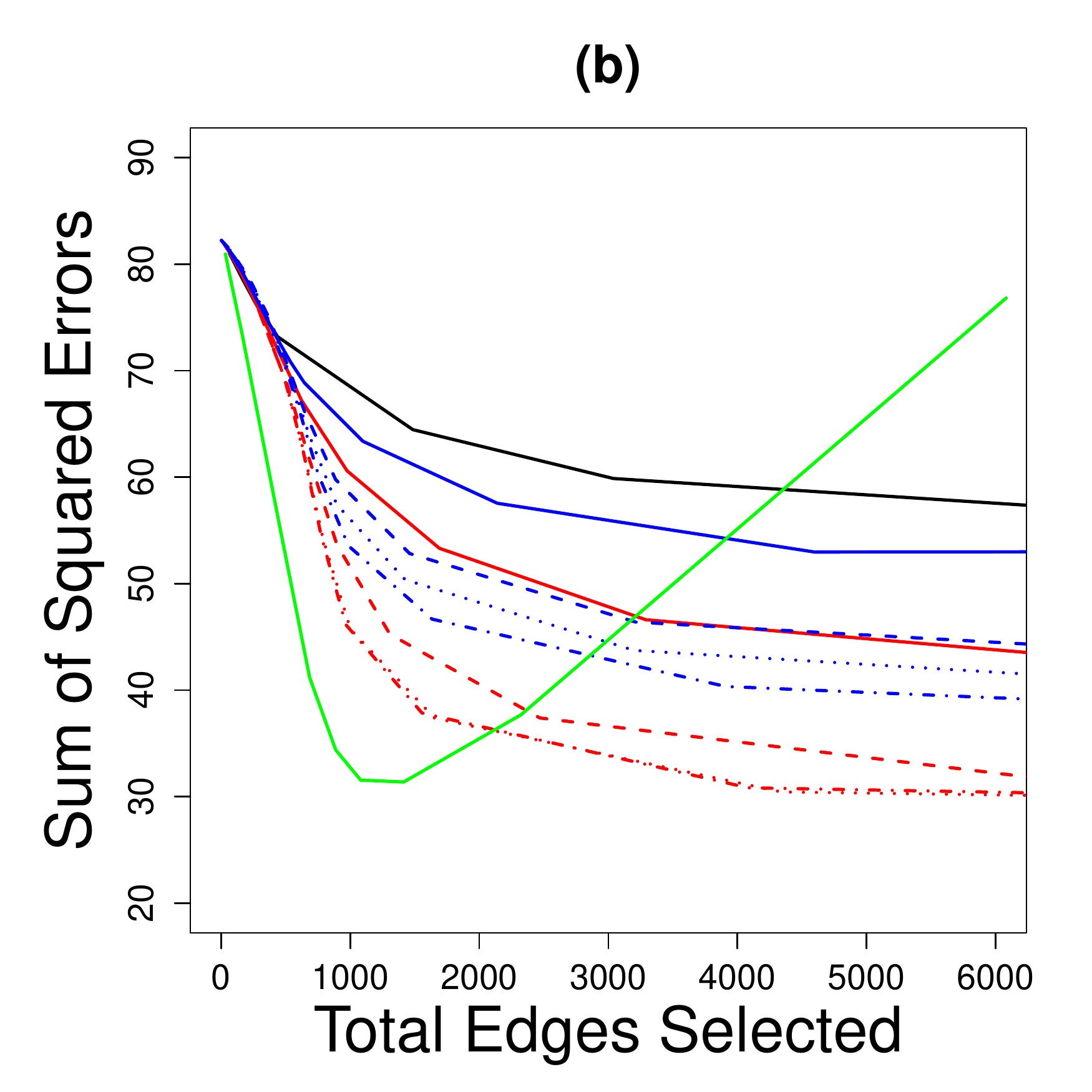}
\includegraphics[width=0.3\linewidth]{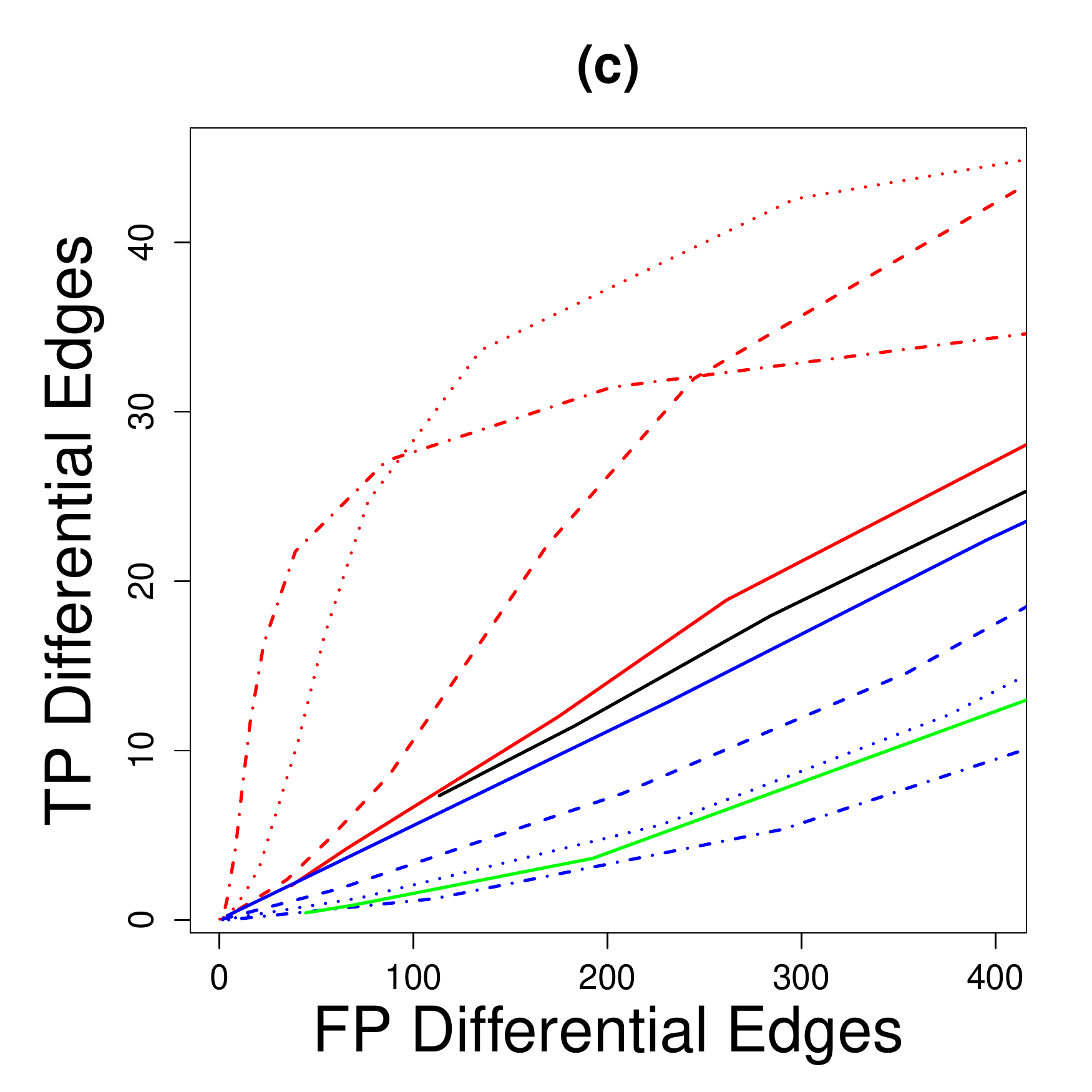}
\includegraphics[width=0.3\linewidth]{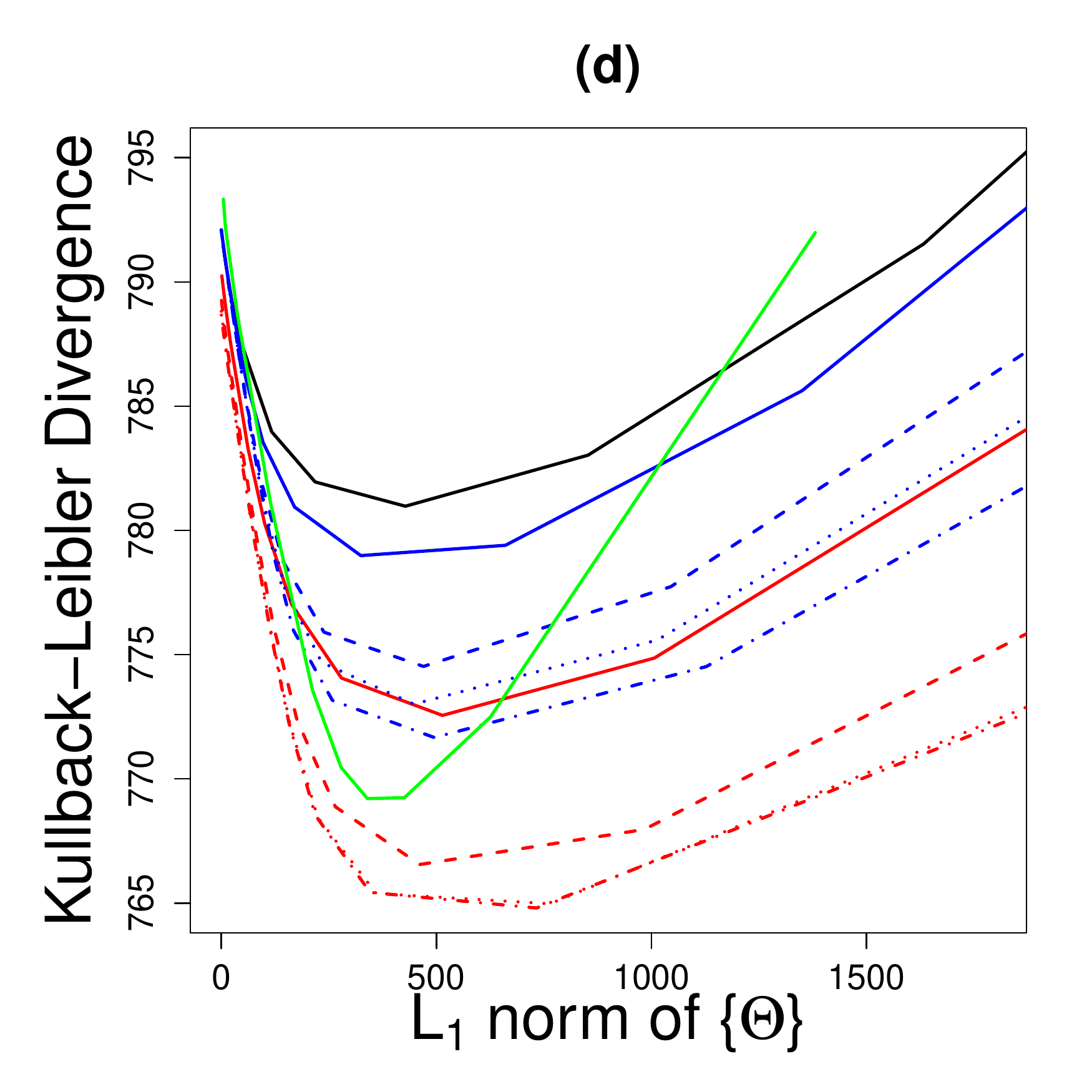}
\includegraphics[width=0.3\linewidth]{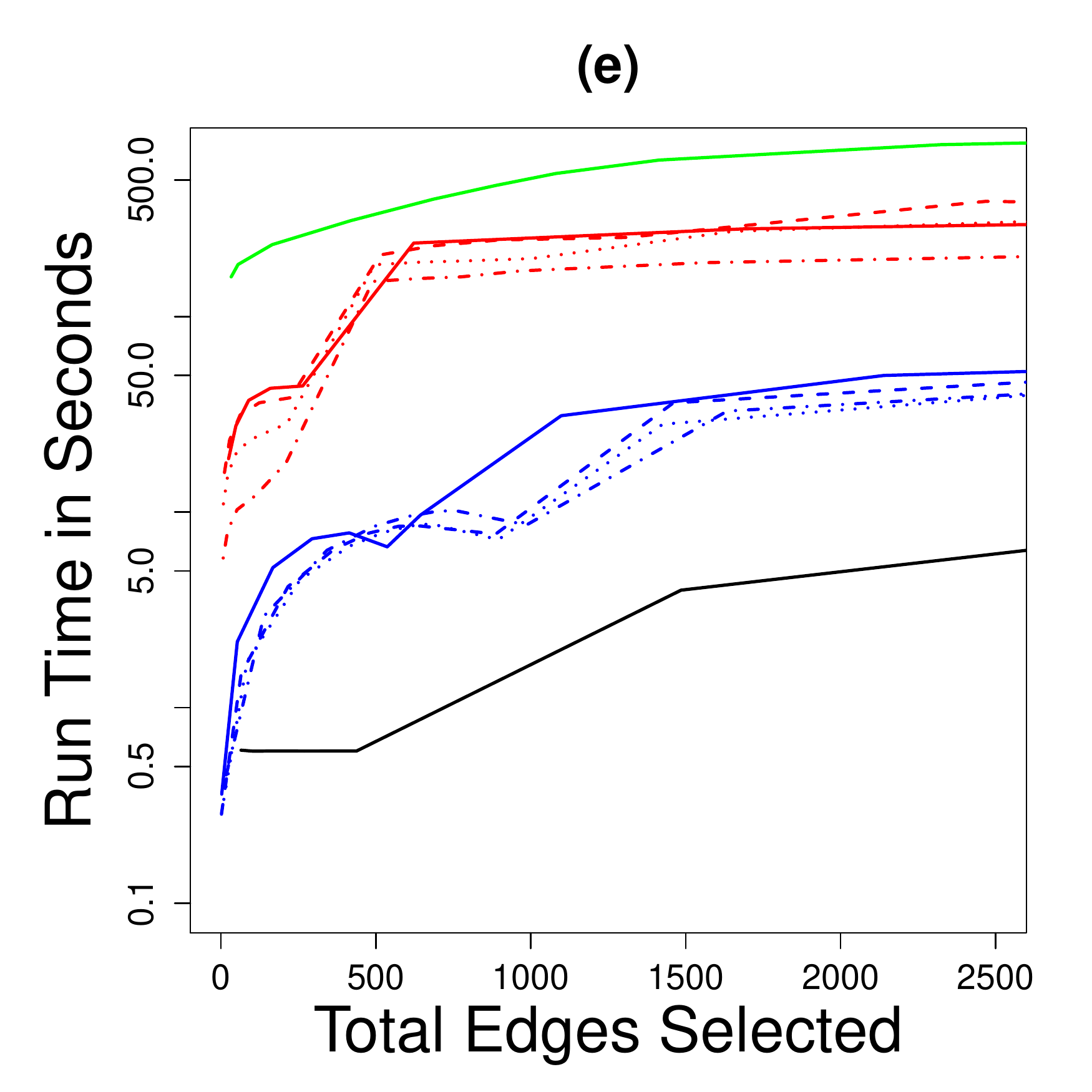}   % use the results from the first K3 sims here!
\includegraphics[width=0.3\linewidth]{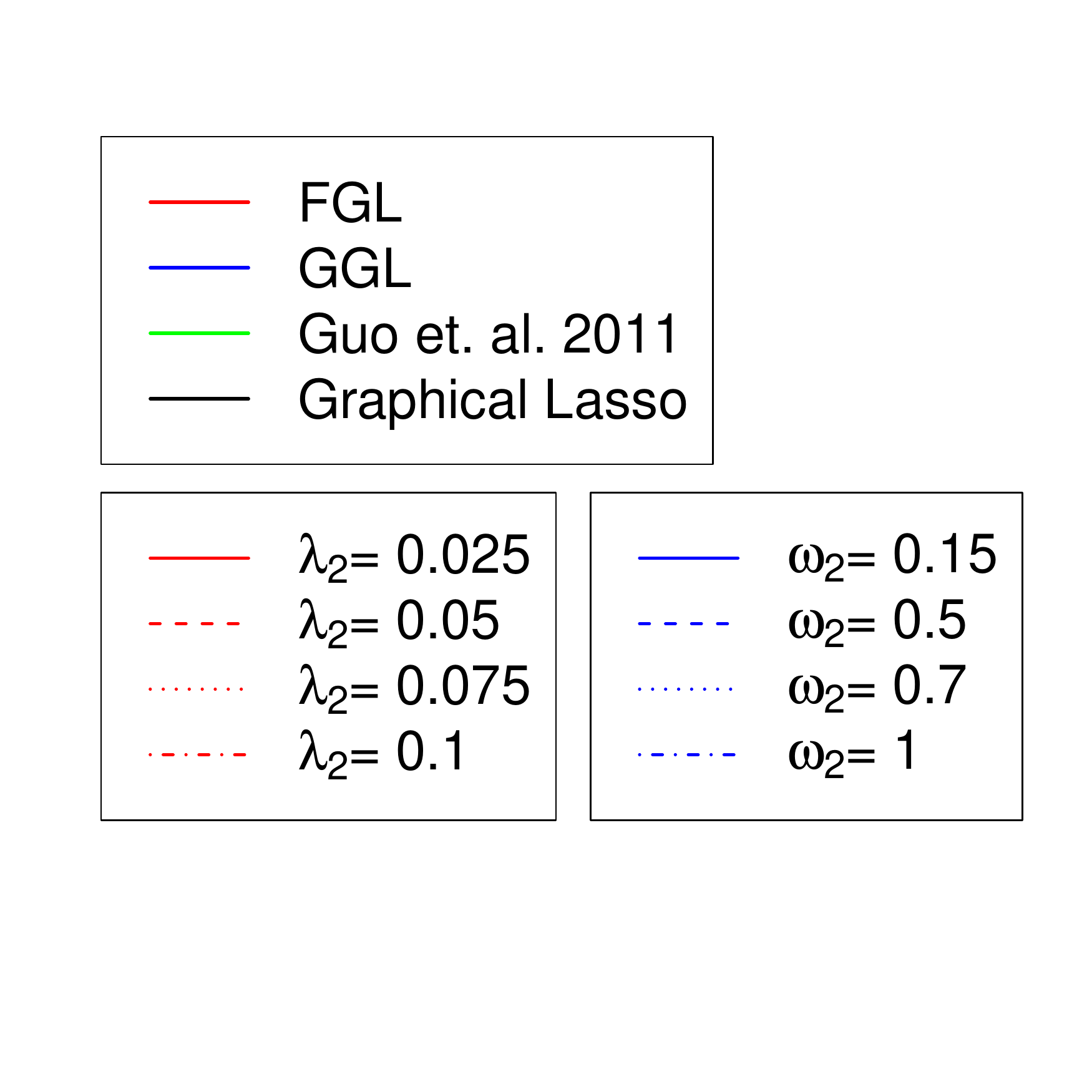}

\caption{ \label{Sixplots}
\it Performance of FGL, GGL, \protect\citet{Guo2011}'s method, and the graphical lasso on simulated data with $150$ observations in each of 3 classes, and $500$ features.
Black lines display models derived using the graphical lasso, green lines display the proposal of \citet{Guo2011}, red lines display FGL, and blue lines display GGL.  
 {\bf (a)}: The number of edges correctly identified to be nonzero (TP Edges) is plotted against the number of  edges incorrectly identified to be nonzero (FP edges).
 {\bf (b)}: The sum of squared errors in edge values is plotted against the total number of edges estimated to be nonzero.
 %We calculate the sum of squared errors as $\sum_{k=1}^K \sum_{i\neq j}(\hat{\bTheta}^{(k)}_{ij} - \siginv_{k,ij})^2$.
 {\bf (c)}: The number of  edges correctly found to have values differing between classes (TP Differential Edges) is plotted against the number of edges incorrectly found
 to have values differing between classes (FP Differential Edges).
 {\bf (d)}: The dKL of the estimated models from the true models is plotted against the $\ell_1$ norm of the off-diagonal entries of the estimated precision matrices.
 {\bf (e)}:  Running time (in seconds) is plotted against the number of non-zero edges estimated. Note the use of a log scale on the $y$-axis.}
\end{figure}

\subsection{Performance as a function of $n$ and $p$}
\label{three}
We now evaluate the effect of sample size $n$ and
dimension $p$ on the performances of FGL and GGL.

\subsubsection{Simulation set-up}

We generate a pair of networks with $p=500$ much as described in Section \ref{setup}, but with $K=2$ instead of $K=3$.
The first network has 10 equal-sized components with power law degree distributions,
and the second network is identical to the first in both edge identity and value,
but with two components removed.

In addition to the 500-feature network
pair, we generate a pair of networks with $p=1000$ features, each of which is block diagonal with  
 $500 \times 500$ blocks
corresponding to two copies of the 500-feature networks just described.
We generate covariance matrices from the networks exactly as described in Section \ref{setup}.

\subsubsection{Simulation results}
For both the $p=500$ and the $p=1000$ networks, we simulate 100
datasets with $n=50$, $n=200$, and $n=500$ samples in each class. We
run FGL with $\lambda_1 = 0.2, \lambda_2 = 0.1$ and GGL with
$\lambda_1 =0.05, \lambda_2 = 0.25$.
%These tuning parameters can be seen in Figure \ref{Sixplots}.
We record in Table \ref{tab:NandP} dKL as well as the sensitivity and false discovery rate
associated with detecting non-zero edges and detecting differential edges.
%We define dKL, true positive and false positive edges, and
%true positive and false positive edge differences as in Section \ref{one}.
In this simulation setting, accuracy of covariance estimation (as
measured by dKL) improves significantly from $n = 50$ to $n = 200$,
and improves only marginally with a further increase to $n = 500$.
Detection of edges  improves throughout the range of
$n$'s sampled: for both FGL and GGL, sensitivity improves slightly
with increased sample size, and FDR decreases dramatically.
Detection of edge differences  is a more difficult problem, for
which FGL performs well.
%though FGL proved useful, with a sensitivity over 0.2 and an FDR around 0.5,
% at $n \geq 200$ and $p=500$, and at all samples sizes when $p=1000$.

\begin{table}
\caption{\label{tab:NandP} \it Performances as a function of $n$ and $p$.
Means over 100 replicates  are shown for dKL,
 and for sensitivity (Sens.) and false discovery rate (FDR) of detection of edges (DE) and differential edge detection (DED).   }
\centering
\small
\fbox{
\begin{tabular}{c|c|c|c|c|c|c|c}
%{\bf FGL} & & & & & \\
%N & P & dKL & Edge detection sensitivity & Edge FDR & Edge differences sensitivity & Differences FDR \\
& $p$ & $n$ & dKL & DE Sens. & DE FDR &DED Sens. & DED FDR\\
% &  &   &     &  sensitivity   &            & sensitivity     & FDR \\
\hline
\multirow{6}{*}{FGL} &  & 50 & 545.1 & 0.502 & 0.966 & 0.262 & 0.996 \\
& 500 & 200 & 517.5 & 0.570  & 0.053 & 0.228  & 0.485  \\
&  & 500 & 516.6 & 0.590 & 0.001 & 0.192 & 0.036  \\
\cline{2-8}
&  & 50 & 1119.3 & 0.600 & 0.970 & 0.245 & 0.998 \\
& 1000 & 200 & 1035.0 & 0.666 & 0.063 & 0.223 & 0.557 \\
&  & 500 & 1033.3 & 0.681 & 0.000 & 0.194 & 0.025 \\
\hline
%{\bf GGL} & & & & & & \\
% n & p & dKL & Edge detection & Edge FDR & Edge differences & Differences FDR \\
%  &   &     &  sensitivity   &          & sensitivity     &    \\\hline
\multirow{6}{*}{GGL} &  & 50 & 549.8 & 0.490 & 0.973 & 0.337 & 0.996  \\
& 500 & 200 & 520.8 & 0.505 & 0.060  & 0.244 & 0.903 \\
&  & 500 & 519.7 & 0.524 & 0.010  & 0.194 & 0.921  \\
\cline{2-8}
&  & 50 & 1127.9 & 0.587 & 0.976 & 0.316 & 0.998 \\
& 1000 & 200 & 1041.7 & 0.615 & 0.061 & 0.239 & 0.908   \\
&  & 500 & 1039.4  & 0.629 & 0.007 & 0.197 & 0.920 \\
\end{tabular}
}
\normalsize
\end{table}

\section{Analysis of lung cancer microarray data}
\label{realstudy}
We applied FGL to a dataset   containing $22,283$ microarray-derived
gene expression measurements from large airway epithelial cells sampled from 97 patients with lung cancer and 90 controls \citep{Spira2007}.
The data are publicly available from the Gene Expression Omnibus \citep{GEO} at accession number GDS2771.
%To ensure we modeled covariances arising more from biology than from noise, we omitted genes with standard deviations in the bottom quintile.
We omitted genes with standard deviations in the bottom $20\%$ since a greater share of their variance is likely attributable to non-biological noise.
The remaining genes were normalized to have mean zero and standard deviation one within each class.
To avoid disparate levels of sparsity between the classes and to prevent the larger class from dominating the estimated networks, we
 weighted each class equally instead of by sample size in  (\ref{loglikjointproposal}).
Since our goal was data visualization and hypothesis generation, we chose a high value for the sparsity tuning parameter, $\lambda_1=0.95$, to yield very sparse network estimates.
We ran FGL with a range of  $\lambda_2$ values in order to identify the edges that differed most strongly, and settled on $\lambda_2=0.005$ as providing the most interpretable results.
Application of Theorem \ref{thm2} revealed that only 278 genes were connected to any other gene
using the chosen tuning parameters.
 Identification of block diagonal structure using Theorem \ref{thm2} and application of the FGL algorithm took less than two minutes.
(Note that this data set is so large that it would be computationally prohibitive to apply the proposal of \citet{Guo2011}!)
FGL estimated 134 edges shared between the two networks, 202 edges present only in the cancer network, and 18 edges present only in the normal tissue network.
The results are displayed in Figure \ref{fig:4}.

\begin{figure}[htp]
\centering
\includegraphics[width=0.4\linewidth]{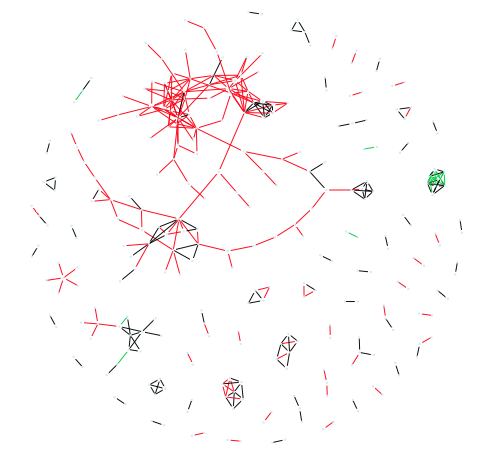}
 \caption{\it{ \label{fig:4}} \it Conditional dependency networks inferred from 17,772 genes in normal and cancerous lung cells. 278 genes have nonzero edges
in at least one of the two networks.
 Black lines denote edges common to both classes. Red and green lines denote tumor-specific and normal-specific edges, respectively. }
\end{figure}

The estimated networks contain many two-gene subnetworks common to both classes, a few small subnetworks, and one large subnetwork specific to tumor cells.
Reassuringly, 45\% of edges, including almost all of the two-gene subnetworks, connect multiple probes for the same gene.
Many other edges connect genes that are obviously related, involved in the same biological process, or even coding for components of the same enzyme.
 Examples include TUBA1B and TUBA1C, PABPC1 and PABPC3, HLA-B and HLA-G, and SERPINB3 and SERPINB4.
Recovery of these pairs suggests that FGL (and other network analysis tools) can generate high-quality hypotheses about gene co-regulation and functional interactions.
This increases our confidence that some of the non-obvious two-gene subnetworks detected in this analysis may merit further investigation. Examples include
 DAZAP2 and TCP1, PRKAR1A and CALM3, and BCLAF1 and SERPB1.
A complete list of subnetworks detected is available in the Supplementary Materials.

The small black and green network in Figure \ref{fig:4} suggests an interesting phenomenon.  It contains multiple probes for two hemoglobin genes, HBA2 and HBB.
In the normal tissue network, the probes for these genes are heavily interconnected both within and between the genes.
In the tumor cells, while edges between HBA2 probes and between HBB probes are preserved, no edges connect the two genes.
The abundance of connections between the two genes in healthy cells and the absence of connections in tumor cells may indicate a possible direction of future investigation.

The most promising results of this analysis arise from the large subnetwork (104 nodes for 84 unique genes) unique to tumor cells.
Many of the subnetwork's genes are involved in constructing ribosomes, including RPS8, RPS23, RPS24, RPS7p11, RPL3, RPL5, RPL10A, RPL14P1, RPL15, RPL17, RPL30 and RPL31.
Other genes in the subnetwork further involve ribosome functioning: SRP14 and SRP9L1 are involved in recruiting proteins from ribosomes into the ER, and NACA inhibits the SRP pathway.
Thus this subnetwork portrays a detailed web of relationships consistent with known biology.
More interestingly, this network also contains two genes in the RAS oncogene family: RAB1A and RAB11A.
Genes in this family have been linked to many types of cancer, and are considered promising targets for therapeutics \citep{Adjei2008}.
These genes' connections with ribosome activity in the tumor samples may indicate a relationship common to an important subset of cancers.
Many other genes belong to this network, each indicating a potentially novel interaction in cancer biology.

\section{Discussion}
\label{discussion}

We have introduced the joint graphical lasso, a method for estimating
sparse inverse covariance matrices on the basis of observations drawn from distinct but related classes.
We employ an ADMM algorithm to solve the joint graphical lasso problem with any convex penalty function,
 and we provide explicit and efficient solutions for two useful penalty functions.
Our algorithm is tractable on very large datasets ($>20,000$ features), and usually converges in seconds for smaller problems (500 features).
Our joint estimation methods outperform competing approaches on a range of simulated datasets.

In the JGL optimization problem (\ref{loglikjointproposal}), the contribution of each class
to the penalized log likelihood is weighted by its size; consequently, the largest class can have  outsize influence on the estimated networks.
By omitting the $n_k$ term in (\ref{loglikjointproposal}), it is  possible to weight the classes equally to prevent a single class from dominating estimation.

We note that FGL and GGL's reliance on two tuning
parameters is a strength rather than a drawback: unlike
the proposal of \citet{Guo2011}, which involves a single tuning parameter that
controls both sparsity and similarity, in performing FGL and GGL one
can vary separately the amount of similarity and sparsity to enforce in the
network estimates.

The joint graphical lasso has potential applications beyond those
discussed in this paper. For instance,  one could use it to shrink
multiple classes' precision matrices towards each other in order
 to define a classifier intermediate between quadratic discriminant analysis (QDA) and linear discriminant analysis (LDA)
\citep{ElemStatLearn}. In fact, a similar approach has been taken in recent work \citep{SimonTibs12}.
In the unsupervised setting, it can be used in the  maximization step of Gaussian model-based clustering in order to reduce the variance associated with estimating a separate covariance matrix
for each cluster.

An R package implementing FGL and GGL will be made available on \verb=CRAN=,\\
 \verb=http://cran.r-project.org/=.

\section*{Acknowledgments}
We thank two anonymous reviewers, an associate editor, and an editor for useful comments that substantially improved this paper;
Noah Simon, Holger Hoefling, Jacob Bien, and Ryan Tibshirani for helpful conversations
 and for sharing with us unpublished results; and Jian Guo and Ji Zhu for
providing software for the proposal in \citet{Guo2011}. We thank Karthik Mohan, Mike Chung, Su-In Lee,  Maryam Fazel, and Seungyeop Han for helpful comments. PD and
PW are supported by NIH grant 1R01GM082802. PW is also supported by
NIH grants P01CA53996 and U24CA086368. DW is supported by NIH grant DP5OD009145.

\small

\bibliographystyle{agsm}
\bibliography{biblio}

\normalsize

\section*{Appendix 1: Modifying JGL to work on the scale of partial correlations}
\label{parcor_extension}

A reviewer suggested that under some circumstances, it may be preferable to encourage the $K$ networks to have shared 
partial correlations, rather than shared precision matrices.
Below, we describe a simple approach for extending our FGL proposal to work on the scale of partial correlations. 
A similar approach can be taken to extend GGL. The extension relies on two insights.
%In order to encourage shared values of the partial correlations rather than shared values of the precision matrices, we rely on two insights.
%penalization on the partial correlation scale.  Our procedure relies on two insights.  First, we note the relationship between the partial correlation and concentration matrices:
\begin{enumerate}
\item $\rho_{ij} = -\frac{\sigma^{ij}}{\sqrt{\sigma^{ii}\sigma^{jj}}}$,
where $\rho_{ij}$ is the true partial correlation between the $i$th
and $j$th features, and where $\sigma^{ij}$ is the $(i,j)$th entry
of the true precision matrix.
%\item The tuning parameter $\lambda_2$ in the JGL optimization problem, which controls the extent to which  the $K$ networks are shrunken towards each other, can be easily modified to take on a different value for each pair of features. That is, we can modify the FGL penalty to take the form
\item The algorithm for solving the FGL optimization problem can easily be modified to make use of the following penalty function:
% The tuning parameters in the JGL optimization problem can be easily modified to take on a different value for each pair of features. That is, we can modify the FGL penalty to take the form
\begin{equation}
P(\sTheta) =  \sum_{k=1}^K \sum_{i \neq j}  \lambda_{1,ij}
|\theta^{(k)}_{ij}| +   \sum_{k < k'} \sum_{i, j}\lambda_{2,ij}
|\theta^{(k)}_{ij}-\theta^{(k')}_{ij}|,
\end{equation}
%and we can write the GGL penalty axs
%\begin{equation}
%P(\sTheta) = \sum_{k=1}^K \sum_{i \neq j} \lambda_{1,ij}
%|\theta^{(k)}_{ij}| +   \sum_{i \neq j}\lambda_{2,ij} \sqrt{\sum_{k=1}^K
%{\theta^{(k)}_{ij}}^2 }.
%\end{equation}
\end{enumerate}
where $\lambda_{t, ij}=\lambda_t/\sqrt{\hat\sigma^{ii}\hat\sigma^{jj}}$, $t=1,\ 2$, and
%Now suppose that we take $\lambda_{1,ij} = \lambda_1 / \sqrt{\hat\sigma^{ii}\hat\sigma^{jj}}$,
where $\hat\sigma^{ii}$ is an estimate of the $i$th diagonal element of the $K$ precision matrices.  
%, and similarly $\lambda_{2,ij} = \lambda_2 / \sqrt{\hat\sigma^{ii}\hat\sigma^{jj}}$. 
(Here, we assume the $K$ precision matrices have shared diagonal elements.)
 The estimate $\{\hat\sigma^{ii}\}$ can be obtained in a number of ways, for instance
by performing the graphical lasso on the samples from all $K$ data sets together.
Then this approach will effectively result in applying a generalized fused lasso penalty to 
the partial correlations for the $K$ classes.\\

%\textbf{Pei and Daniela: I changed this so we're modifying both penalties instead of only lambda2.  I think it makes more sense this way, especially for GGL.  Do you agree?}  %fff

\section*{Appendix 2: Proofs of Theorems \ref{thm2} and \ref{thm3}}
\subsection*{Preliminaries to Proofs of Theorems \ref{thm2} and \ref{thm3}}

\noindent We begin with  a few comments on subgradients.
The subgradient of $|\theta_{ij}^{(k)}|$ with respect to $\theta_{ij}^{(k)}$ equals
\[
\left\{ \begin{array}{lll}
1 & \mbox{if $\theta_{ij}^{(k)}>0$} \\
-1 & \mbox{if $\theta_{ij}^{(k)}<0$} \\
a  & \mbox{if $\theta_{ij}^{(k)}=0$} \\
\end{array},
\right.
\]
for some $a \in [-1,1]$.
The subgradient of $ |\theta_{ij}^{(k)} - \theta_{ij}^{(k')}|$ with respect to $(\theta_{ij}^{(k)}, \theta_{ij}^{(k')})$ for $k \neq k'$ equals $(d, -d)$, where
\[
d = \left\{ \begin{array}{lll}
1 & \mbox{if $\theta_{ij}^{(k)} > \theta_{ij}^{(k')}$} \\
-1 & \mbox{if $\theta_{ij}^{(k)} < \theta_{ij}^{(k')}$} \\
a  & \mbox{if $\theta_{ij}^{(k)} = \theta_{ij}^{(k')}$} \\
\end{array},
\right.
\]
for some $a \in [-1,1]$.
Finally, the subgradient of $\sqrt{\sum_{k=1}^K (\theta_{ij}^{(k)})^2}$ with respect to $(\theta_{ij}^{(1)},\ldots, \theta_{ij}^{(K)})$  is given by
\[
\left\{ \begin{array}{ll}
 (\theta_{ij}^{(1)}, \ldots,  \theta_{ij}^{(K)}) / \sum_{k=1}^K (\theta_{ij}^{(k)})^2 & \mbox{if $\sum_{k=1}^K (\theta_{ij}^{(k)})^2  > 0$} \\
(\Upsilon_{1,ij}, \ldots, \Upsilon_{K,ij}) & \mbox{if $\theta_{ij}^{(1)} = \ldots = \theta_{ij}^{(K)} = 0$} \\
\end{array},
\right.
\]
for some $\Upsilon_{1,ij},\ldots,\Upsilon_{K,ij}$ such that  $\sum_{k=1}^K \Upsilon_{k,ij}^2 \leq 1$.

%\subsection*{Proof of Theorem \ref{thm2}}

To prove Theorem \ref{thm2}, we will use the following lemma.

\begin{lemma}
\label{lemma1}
The following two sets of conditions are equivalent:
\begin{list}{}{}
\item{(A):} $|n_1 S_1| \leq \lambda_1 + \lambda_2$,
 $|n_2 S_2| \leq \lambda_1 + \lambda_2$, and
 $|n_1 S_1 + n_2 S_2| \leq 2\lambda_1$.
\item{(B):} There exist $\Gamma_1, \Gamma_2, \Upsilon \in [-1,1]$ such that
$-n_1 S_1 - \lambda_1 \Gamma_{1} - \lambda_2 \Upsilon = 0$, and  $-n_2 S_2- \lambda_1 \Gamma_{2} + \lambda_2 \Upsilon = 0$.
\end{list}
\end{lemma}

\begin{proof}
We will begin by proving that $(B)$ implies $(A)$, and will then prove that $(A)$ implies $(B)$.
%{\color{white}.}

\noindent \emph{Proof that {(B)} $\Rightarrow$ {(A)}:}  \\

\noindent First of all, $-n_1S_1 - \lambda_1 \Gamma_1 - \lambda_2 \Upsilon = 0$ implies that $|n_1 S_1| \leq \lambda_1 + \lambda_2$, since $\Gamma_1, \Upsilon \in [-1,1]$.
Similarly, $-n_2 S_2 - \lambda_1 \Gamma_2 + \lambda_2 \Upsilon = 0$ implies that $|n_2 S_2| \leq \lambda_1 + \lambda_2$.
Finally, summing the two equations in $(B)$ reveals that
 $n_1 S_1 + n_2 S_2 = - \lambda_1 (\Gamma_1 + \Gamma_2)$, which implies that
 $|n_1 S_1 + n_2 S_2| \leq 2 \lambda_1$.  \\

\noindent  \emph{Proof that {(A)} $\Rightarrow$ {(B)}:} \\

\noindent Without loss of generality, assume that $n_1 S_1 \geq n_2 S_2$.
We split the proof into two cases.
\begin{enumerate}
\item {\emph{Case 1:}} $n_1 S_1 - n_2 S_2  < 2 \lambda_2$.

Let $\Gamma_1 = \Gamma_2 = \frac{-n_1 S_1 - n_2 S_2}{2 \lambda_1}$, and
    $\Upsilon = \frac{-n_1 S_1 + n_2 S_2 }{2 \lambda_2}$.

First, note that by (A), we know that $|n_1 S_1+n_2 S_2| \leq 2\lambda_1$. Therefore, $\Gamma_1, \Gamma_2 \in [-1,1]$.
Second, note that Case 1's assumption that $n_1 S_1 - n_2 S_2  < 2 \lambda_2$  implies that $\Upsilon \in [-1,1]$.
Finally, we see by inspection that $-n_1 S_1 - \lambda_1 \Gamma_{1} - \lambda_2 \Upsilon = 0$, and  $-n_2 S_2- \lambda_1 \Gamma_{2} + \lambda_2 \Upsilon = 0$.

\item {\emph{Case 2:}} $n_1 S_1 - n_2 S_2 \geq  {2 \lambda_2}$.

Let $\Gamma_1 = \frac{-n_1 S_1+\lambda_2}{\lambda_1}$,
$\Gamma_2 = \frac{-n_2 S_2-\lambda_2}{\lambda_1}$, and $\Upsilon = -1$.
Then, by inspection, $-n_1 S_1 - \lambda_1 \Gamma_{1} - \lambda_2 \Upsilon = 0$, and  $-n_2 S_2- \lambda_1 \Gamma_{2} + \lambda_2 \Upsilon = 0$.

It remains to  show that $\Gamma_1, \Gamma_2, \Upsilon \in [-1,1]$.
Trivially, $\Upsilon = -1 \in [-1,1]$.  From our assumption that $|n_1 S_1| \leq \lambda_1 + \lambda_2$, we know that
 $-1 \leq \Gamma_1$. Moreover, by the assumptions that
${n_1 S_1- n_2 S_2} \geq {2 \lambda_2} $ and $|n_1 S_1 + n_2 S_2| \leq 2\lambda_1$, we have that
\begin{equation}
 \Gamma_1 = \frac{-n_1 S_1 +\lambda_2}{\lambda_1} \leq \frac{-n_1 S_1 +\lambda_2 \left(\frac{n_1 S_1-n_2 S_2}{2 \lambda_2} \right)}{\lambda_1}
=\frac{-n_1 S_1 -n_2 S_2  }{2 \lambda_1} \leq 1.
 \end{equation}
Therefore $\Gamma_1 \in [-1,1]$.

By the assumption that $| n_2 S_2| \leq \lambda_1 + \lambda_2$, we know that
$\Gamma_2 =  \frac{-n_2 S_2 -\lambda_2}{\lambda_1} \leq 1$.
From the assumptions that $n_1 S_1- n_2 S_2 \geq 2 \lambda_2$ and $|n_1 S_1 + n_2 S_2 | \leq 2\lambda_1$, we have that
\begin{equation}
\Gamma_2 = \frac{-n_2 S_2 - \lambda_2}{\lambda_1} \geq \frac{-n_2 S_2 -\lambda_2 \left( \frac{n_1 S_1 - n_2 S_2}{2 \lambda_2} \right)}{\lambda_1} = \frac{-n_1 S_1 - n_2  S_2}{2 \lambda_1} \geq -1.
\end{equation}
Therefore $\Gamma_2 \in [-1,1]$.
\end{enumerate}
Thus we conclude {(A)} $\Rightarrow$ {(B)}, and our proof of Lemma \ref{lemma1} is complete.
\end{proof}
%\end{lemma}

We will make use of the following lemma in order to prove   Theorem \ref{thm3}.
\begin{lemma}
\label{lemma2}
The following two conditions are equivalent:
\begin{list}{}{}
\item{(A):}  There exist scalars $a_1,\ldots,a_K$  such that $\sum_{k=1}^K a_k^2 \leq 1$ and $n_k |S_k | \leq \lambda_1 + \lambda_2 a_k$ for all $k=1,\ldots,K$.
% $n_2 |S_2 | \leq \lambda_1 + \lambda_2 \sqrt{1-a^2}$, for some $a \in [0,1]$.
 \item{(B):} There exist scalars $\Gamma_1, \ldots, \Gamma_K \in [-1,1]$ and  $\Upsilon_1, \ldots, \Upsilon_K$ such that $\sum_{k=1}^K \Upsilon_k^2 \leq  1$ and
 %, $\Upsilon_1^2 + \Upsilon_2^2 \leq 1$,  such that $-n_1 S_1 - \lambda_1 \Gamma_1 - \lambda_2 \Upsilon_1 =0$ and $-n_2 S_2 - \lambda_1 \Gamma_2 -
 %\lambda_2 \Upsilon_2 = 0$.
 $n_k S_k + \lambda_1 \Gamma_k + \lambda_2 \Upsilon_k=0$ for $k=1,\ldots,K$.
\end{list}
\end{lemma}

\begin{proof} We will begin by proving that $(B)$ implies $(A)$, and will then show that $(A)$ implies $(B)$.\\
%{\color{white}.}

%\vspace{2mm}

\noindent \emph{Proof that ${(B)} \Rightarrow {(A)}$: } \\

\noindent By $(B)$, $n_k |S_k| = | \lambda_1 \Gamma_k + \lambda_2 \Upsilon_k|  \leq \lambda_1 |\Gamma_k| + \lambda_2 |\Upsilon_k| \leq \lambda_1 + \lambda_2 |\Upsilon_k| $. Letting $a_k = |\Upsilon_k|$, the result holds.\\

%It is clear that  $|\Upsilon_2| \leq \sqrt{1-\Upsilon_1^2}$.
%Let $|\Upsilon_1|$ = a.  Then, since $-n_1 S_1 - \lambda_1 \Gamma_1 - \lambda_2 \Upsilon_1 = 0$, it follows that  $n_1 |S_1| \leq \lambda_1 + \lambda_2 a$.
%Similarly, $-n_2 S_2 - \lambda_1 \Gamma_2 - \lambda_2 \Upsilon_2 = 0$ implies that  $n_2 |S_2 | \leq \lambda_1 + \lambda_2 |\Upsilon_2| =  \lambda_1 + \lambda_2 \sqrt{1-%a^2}$. \\

\noindent \emph{Proof that ${(A)} \Rightarrow {(B)}$:}  \\

Let $\Gamma_k$ and $\Upsilon_k$ take the following forms, for $k=1,\ldots,K$:
\begin{eqnarray}
\Gamma_k = \begin{cases}
-1 & \mathrm{\;\; if \;\;} n_k S_k > \lambda_1\\
-n_k S_k / \lambda_1  &\mathrm{\;\; if \;\;} -\lambda_1 < n_k S_k <  \lambda_1\\
1& \mathrm{\;\; if \;\;} n_k S_k < - \lambda_1
\end{cases}\\
\Upsilon_k = \begin{cases}
(-n_k S_k + \lambda_1)/\lambda_2 & \mathrm{\;\; if \;\;} n_k S_k > \lambda_1\\
0  &\mathrm{\;\; if \;\;} -\lambda_1 < n_k S_k <  \lambda_1\\
(-n_k S_k - \lambda_1)/\lambda_2 & \mathrm{\;\; if \;\;} n_k S_k < - \lambda_1
\end{cases}
\end{eqnarray}
First of all, we note by inspection that $\Gamma_k \in [-1,1]$ and that $n_k S_k + \lambda_1 \Gamma_k + \lambda_2 \Upsilon_k=0$ for $k=1,\ldots,K$. It remains to show that
$\sum_{k=1}^K \Upsilon_k^2 \leq 1$. Specifically, we will show that $\Upsilon_k^2 \leq a_k^2$ for $k = 1,\ldots,K$.
To see why this is the case, note that if $-\lambda_1 < n_k S_k <  \lambda_1$ then $0 = \Upsilon_k^2 \leq a_k^2$.
And if $n_k S_k > \lambda_1$ or $n_k S_k < -\lambda_1$, then $\Upsilon_k^2 = \left( \frac{n_k |S_k| - \lambda_1}{\lambda_2} \right)^2 \leq a_k^2$.
%\noindent Without loss of generality, assume $S_1 > 0$ and $S_2 > 0$.
%We will consider four distinct cases.
%For each case, we will identify $\Gamma_1, \Gamma_2 \in [-1,1]$ and $\Upsilon_1, \Upsilon_2$ satisfying $\sqrt{\Upsilon_1^2 + \Upsilon_2^2} \leq 1$
% for which
% $-n_1 S_1 - \lambda_1 \Gamma_1 - \lambda_2 \Upsilon_1 = 0$
%and
% $-n_2 S_2 - \lambda_1 \Gamma_2 - \lambda_2 \Upsilon_2 = 0$.
%The confirmation that these conditions on $\Gamma_1,\Gamma_2, \Upsilon_1,$ and  $\Upsilon_2$ hold in each case is trivial, and we omit it.\\
%
%\begin{enumerate}
%\item \emph{ Case 1:} $n_1 S_1> \lambda_1, n_2 S_2 > \lambda_1$.\\
%
%Take $\Gamma_1 = \Gamma_2 = -1$, and $\Upsilon_1 = \frac{\lambda_1-n_1 S_1}{\lambda_2}, \Upsilon_2 = \frac{\lambda_1-n_2 S_2}{\lambda_2}$.
%
%\item \emph {Case 2:} $n_1 S_1 \leq \lambda_1, n_2 S_2 > \lambda_1$.\\
%
%Take $\Gamma_1 = \frac{-n_1 S_1}{\lambda_1}, \Gamma_2=-1, \Upsilon_1=0$, and $\Upsilon_2 = \frac{\lambda_1-n_2 S_2}{\lambda_2}.$
%
%\item \emph{Case 3:} $n_1 S_1 > \lambda_1, n_2 S_2 \leq \lambda_1$.\\
%
%Take $\Gamma_1 = -1, \Gamma_2 = \frac{-n_2 S_2}{\lambda_1}, \Upsilon_1= \frac{\lambda_1-n_1 S_1}{\lambda_2}$, and $\Upsilon_2 = 0$.
%
%\item \emph{Case 4:} $n_1 S_1 \leq \lambda_1, n_2 S_2 \leq \lambda_1$.\\
%
%Take $\Gamma_1 = \frac{-n_1 S_1}{\lambda_1}, \Gamma_2=\frac{- n_2 S_2}{\lambda_1}, \Upsilon_1=0$, and $\Upsilon_2 = 0$.
%\end{enumerate}
%
%\noindent The subgradient values detailed above show that {(A)} $\Rightarrow$ {(B)} in all four cases, completing our proof.
\end{proof}

\subsection*{Proof of Theorem \ref{thm2}}
%We now proceed with the proof of Theorem \ref{thm2}.
We first consider the claim for the case $K=2$. By the Karush-Kuhn-Tucker \cite[KKT; see e.g.][]{BoydConvex} conditions, a necessary
and sufficient set of conditions for $\sTheta$ to be the solution to the JGL problem  is that
\begin{eqnarray}
0 &=& n_1 ({\bf \Theta}^{(1)})^{-1} - n_1 {\bf S}^{(1)} - \lambda_1 {\bf\Gamma}_1 - \lambda_2 \bf\Upsilon \nonumber \\
0 &=& n_2 ({\bf \Theta}^{(2)})^{-1} - n_2 {\bf S}^{(2)} - \lambda_1 {\bf\Gamma}_2 + \lambda_2 \bf\Upsilon,
\label{awful1st}
\end{eqnarray}
%In (\ref{awful1st}), 
where ${\Gamma}_{1,ij}$ is the subgradient of $|\theta_{ij}^{(1)}|$ with respect to $\theta_{ij}^{(1)}$,
${\Gamma}_{2,ij}$ is the subgradient of  $|\theta_{ij}^{(2)}|$ with respect to $\theta_{ij}^{(2)}$,
and
  $\Upsilon_{ij}$ is the subgradient of $|\theta_{ij}^{(1)}-\theta_{ij}^{(2)}|$ with respect to $\theta_{ij}^{(1)}$.

Let $C_1$ and $C_2$ be a partition of the $p$ variables into two nonoverlapping sets, with
$C_1 \cap C_2 = \emptyset$, $C_1 \cup C_2 = \{ 1, \ldots, p \}$.
Consider the matrices
\begin{equation}
\bTheta^{(1)} = \left( \begin{array}{cc} \bTheta^{(1)}_{1} & 0 \\ 0 & \bTheta^{(1)}_{2} \end{array} \right), \;\;  \bTheta^{(2)} = \left( \begin{array}{cc} \bTheta^{(2)}_{1} & 0 \\ 0 & \bTheta^{(2)}_{2} \end{array} \right),
\label{solnxxx}
\end{equation}
where $\bTheta^{(1)}_1$ and $\bTheta^{(2)}_1$ solve the JGL problem on the features in $C_1$, and $\bTheta^{(1)}_2$ and $\bTheta^{(2)}_2$ solve the JGL problem on the features in $C_2$.
By inspection of (\ref{awful1st}), $\bTheta^{(1)}$ and $\bTheta^{(2)}$ solve the entire JGL optimization problem
 if and only if for all $i \in C_1$, $j \in C_2$, there exist $\Gamma_{1,ij}, \Gamma_{2,ij}, \Upsilon_{ij} \in [-1,1]$ such that
\begin{eqnarray}
-n_1 S_{ij}^{(1)} - \lambda_1 \Gamma_{1,ij} - \lambda_2 \Upsilon_{ij} &=& 0 \nonumber \\
-n_2 S_{ij}^{(2)} - \lambda_1 \Gamma_{2,ij} + \lambda_2 \Upsilon_{ij} &=& 0.
\end{eqnarray}
 Therefore, by Lemma \ref{lemma1}, the proof of the claim for the case $K=2$ is complete.

 The derivation of the necessary condition for the case $K>2$ is simple and we omit it here.

\subsection*{Proof of Theorem \ref{thm3}}
We note that Theorem \ref{thm3}'s condition (\ref{thm3condition}) is equivalent to the following:
\begin{equation}
%\sum_{k=1}^K (|n_kS^{(k)}_{ij}|-\lambda_1)_+^2 \leq \lambda_2^2
|n_k S^{(k)}_{ij} | \leq \lambda_1 + \lambda_2 a_{ij,k} \mathrm{\;\; for\; all \;\;} i \in C_1, \; j \in C_2, \; k=1,\ldots,K
%\mathrm{\;\; for \; all \;\;} k=1,\ldots,K, i\in C_k, \mathrm{\; and \;} j \in C_2,
\label{thm3condition_dw}
\end{equation}
where $a_{ij,1}, \ldots, a_{ij,K}$ are scalars that satisfy $\sum_{k=1}^K a_{ij,k}^2 \leq 1$.
We will prove that (\ref{thm3condition_dw}) is necessary and sufficient for the variables in $C_1$ to be completely disconnected from those in $C_2$ in each of the resulting network estimates. \\

 By the KKT conditions, a necessary
and sufficient set of conditions for $\sTheta$ to be the solution to the JGL problem  is that
\begin{equation}
0 = n_k ({\bf \Theta}^{(k)})^{-1} - n_k {\bf S}^{(k)} - \lambda_1 {\bf\Gamma}_k - \lambda_2 {\bf\Upsilon}_k% \nonumber \\
%0 &=& n_2 ({\bf \Theta}^{(2)})^{-1} - n_2 {\bf S}^{(2)} - \lambda_1 {\bf\Gamma}_2 - \lambda_2 {\bf\Upsilon}_2.
\label{awful}
\end{equation}
for $k=1,\ldots,K$.
In (\ref{awful}), ${\Gamma}_{k,ij}$ is the subgradient of $|\theta_{ij}^{(k)}|$ with respect to $\theta_{ij}^{(k)}$,
%${\Gamma}_{,ij}$ is the subgradient of  $|\theta_{ij}^{(2)}|$,
and $({\Upsilon}_{1,ij}, \ldots,{\Upsilon}_{K,ij})$ is the subgradient of
$\sqrt{\sum_{k=1}^K (\theta_{ij}^{(k)})^2}$ with respect to $(\theta_{ij}^{(1)}, \ldots,\theta_{ij}^{(K)})$.

Let $C_1$ and $C_2$ be a partition of the $p$ variables into two nonoverlapping sets, with
$C_1 \cap C_2 = \emptyset$, $C_1 \cup C_2 = \{ 1, \ldots, p \}$.
Consider the matrices of the form
\begin{equation}
\bTheta^{(k)} = \left( \begin{array}{cc} \bTheta^{(k)}_{1} & 0 \\ 0 & \bTheta^{(k)}_{2} \end{array} \right)%,  \ldots,   \bTheta^{(K)} = \left( \begin{array}{cc} \bTheta^{(K)}_{1} & 0 \\ 0 & \bTheta^{(K)}_{2} \end{array} \right),
\label{solnxxx}
\end{equation}
for $k=1,\ldots,K$,
where $\bTheta^{(1)}_1, \ldots, \bTheta^{(K)}_1$ solve the JGL problem on the features in $C_1$, and $\bTheta^{(1)}_2, \ldots, \bTheta^{(K)}_2$ solve the JGL problem on the features in $C_2$. % and $\bTheta^{(2)}_2$.
By inspection of (\ref{awful}), $\bTheta^{(1)},\ldots,\bTheta^{(K)}$ solve the entire JGL optimization problem
 if and only if for all $i \in C_1$, $j \in C_2$, there exist $\Gamma_{1,ij}, \ldots, \Gamma_{K,ij} \in [-1,1]$ and $\Upsilon_{1,ij},\ldots, \Upsilon_{K,ij}$ satisfying $\sum_{k=1}^K \Upsilon_{k,ij}^2  \leq 1$ such that
\begin{equation}
-n_k S_{ij}^{(k)} - \lambda_1 \Gamma_{k,ij} - \lambda_2 \Upsilon_{k,ij} = 0
%-n_2 S_{ij}^{(2)} - \lambda_1 \Gamma_{2,ij} - \lambda_2 \Upsilon_{2,ij} &=& 0.
\end{equation}
 Therefore, by Lemma \ref{lemma2}, the proof is complete.

\section*{Appendix 3: Additional simulations for two-class datasets}
\label{moresims}

%\subsection*{Simulations with two classes and a block diagonal power law network}
We first present results for a simulation study similar to the one in Section \ref{K3sims}, but with only two classes.
%
%\subsubsection*{Simulation set-up}
%
Taking an approach similar to the one described in Section \ref{K3sims},
we defined two networks with $p=500$ features belonging to ten equally
sized unconnected subnetworks, each with a power law degree
distribution.
Of the ten subnetworks, eight have the same structure and edge values
in both classes, and two are present in only one class.
Class 1's network has 490 edges, 94 of which are not present in class 2.
We generated covariance matrices as described in Section \ref{setup}.
Again, we simulated 100 datasets with $n=150$ observations per class.
The results shown in Figure \ref{Sixplots.K2} are similar to the results in Section \ref{one}.

%\subsubsection*{Simulation results}
%Figure \ref{Sixplots.K2} shows the results, averaged over the 100 data sets.
%In each plot, the lines for FGL and for GGL connect results obtained with a single value of the similarity tuning parameters $\lambda_2$ and $\omega_2$.
%
%Figure \ref{Sixplots.K2}a plots the number of edges correctly estimated
%to be nonzero (TP Edges) against the number of edges incorrectly estimated
%to be nonzero (FP Edges).
%Figure \ref{Sixplots.K2}b plots the sum of squared errors (SSE) of the estimated edges against
%the total number of edges estimated.  SSE is calculated as $\sum_{k=1}^K \sum_{i\neq j}(\hat{\bTheta}^{(k)}_{ij} - \siginv_{k,ij})^2$.
%Figure \ref{Sixplots.K2}c displays the number of edges correctly
%identified as having differing values in the two networks (TP Differential Edges) against
%the number of edges incorrectly identified as having differing values in the
%two networks (FP Differential Edges).
%Figure \ref{Sixplots.K2}d displays the the sum of the dKL's of the estimated distributions from the true distributions, as a function of  the $\ell_1$ norm of the estimated precision %matrices.
%Figure \ref{Sixplots.K2}e compares the methods' running times in seconds.

\begin{figure}[htp]
\centering
\includegraphics[width=0.25\linewidth]{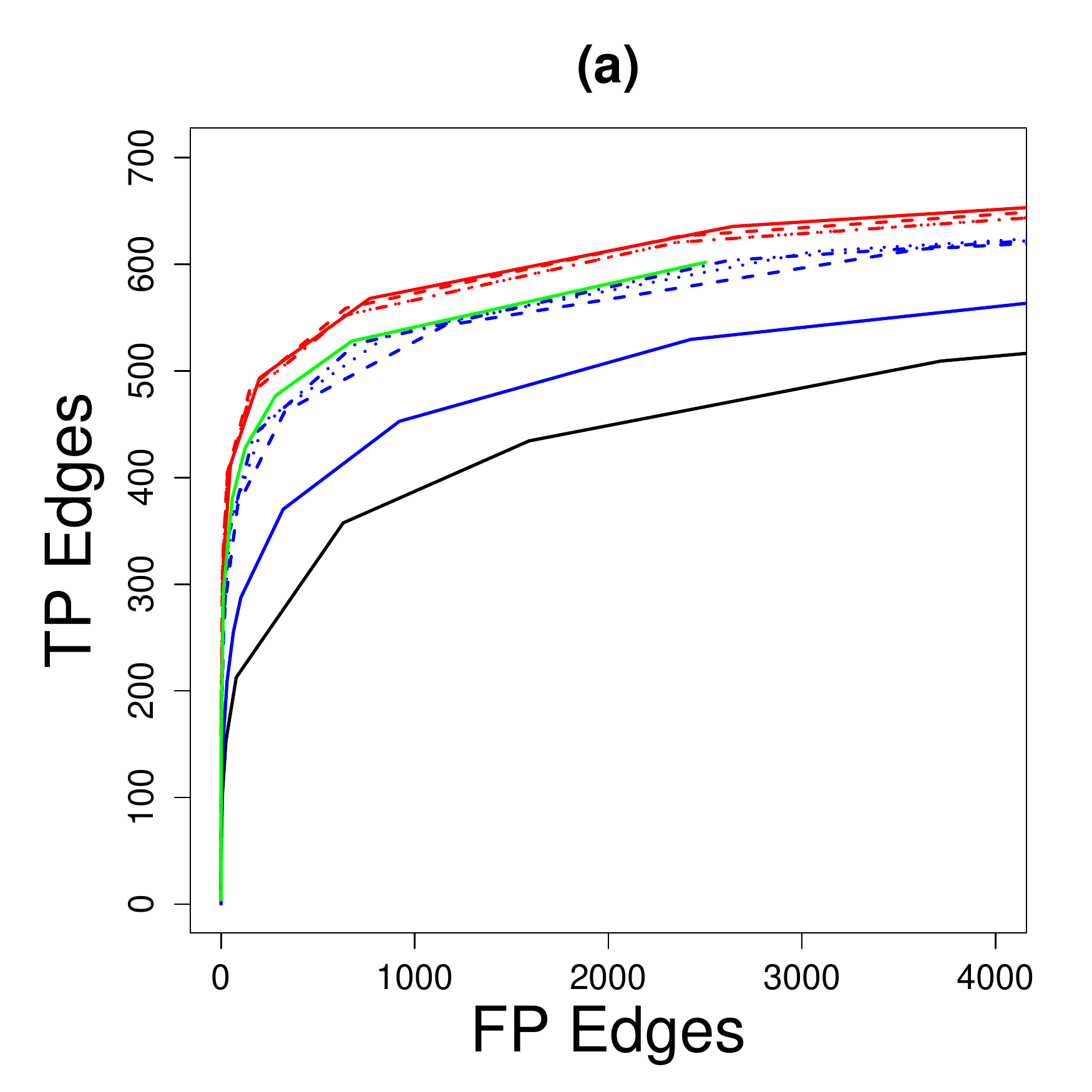}   %note: these files don't exist yet!
\includegraphics[width=0.25\linewidth]{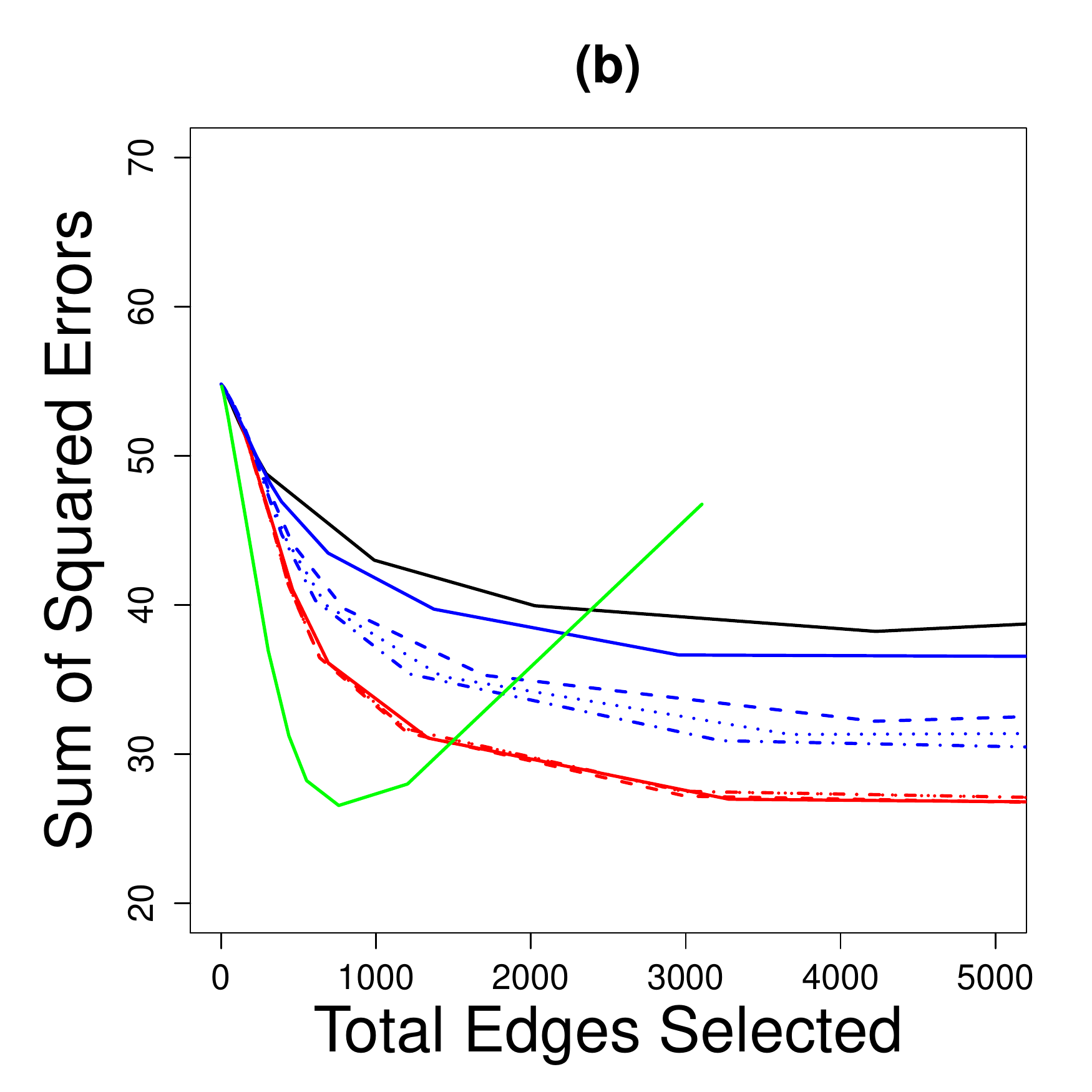}
\includegraphics[width=0.25\linewidth]{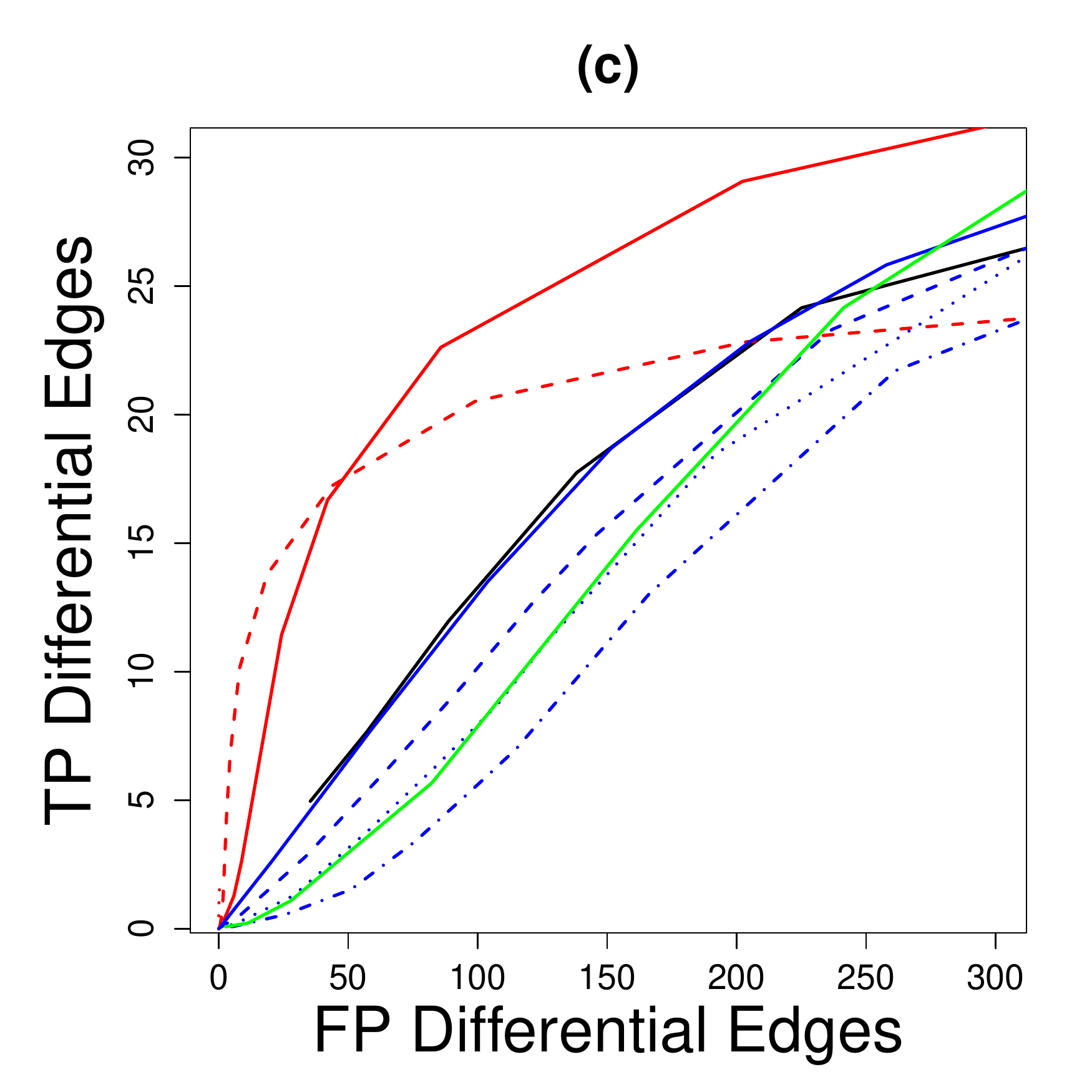}
\includegraphics[width=0.25\linewidth]{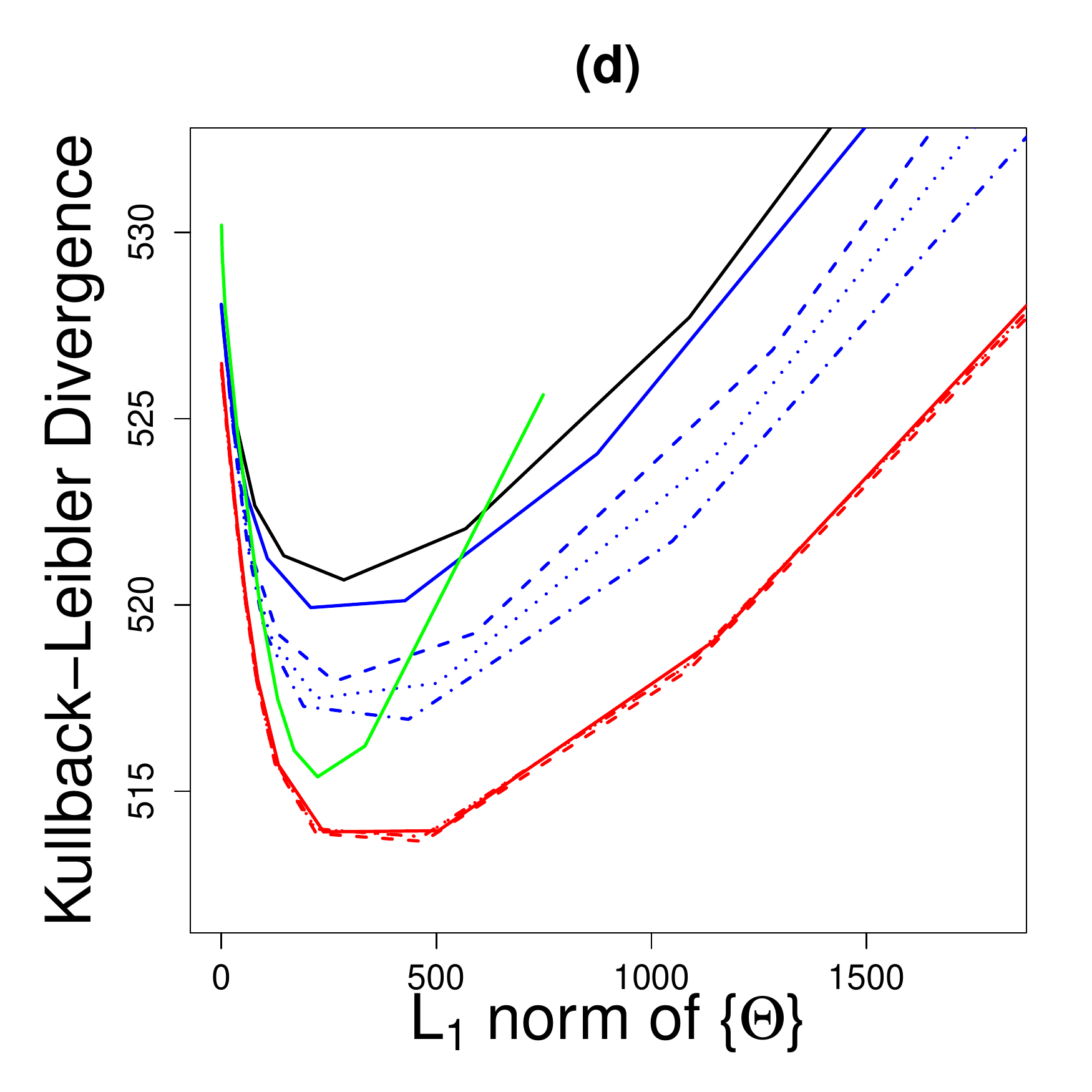}
\includegraphics[width=0.25\linewidth]{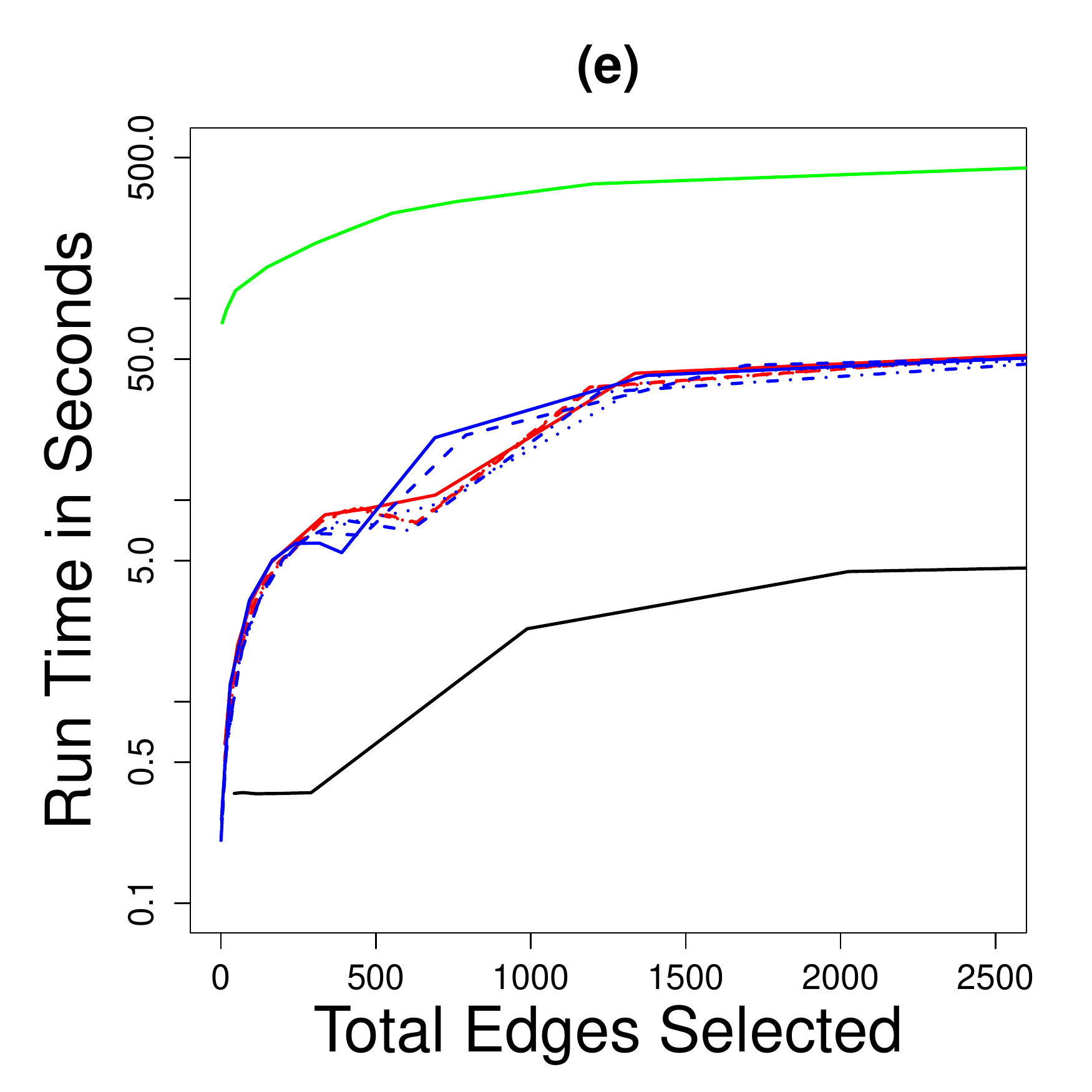}
\includegraphics[width=0.25\linewidth]{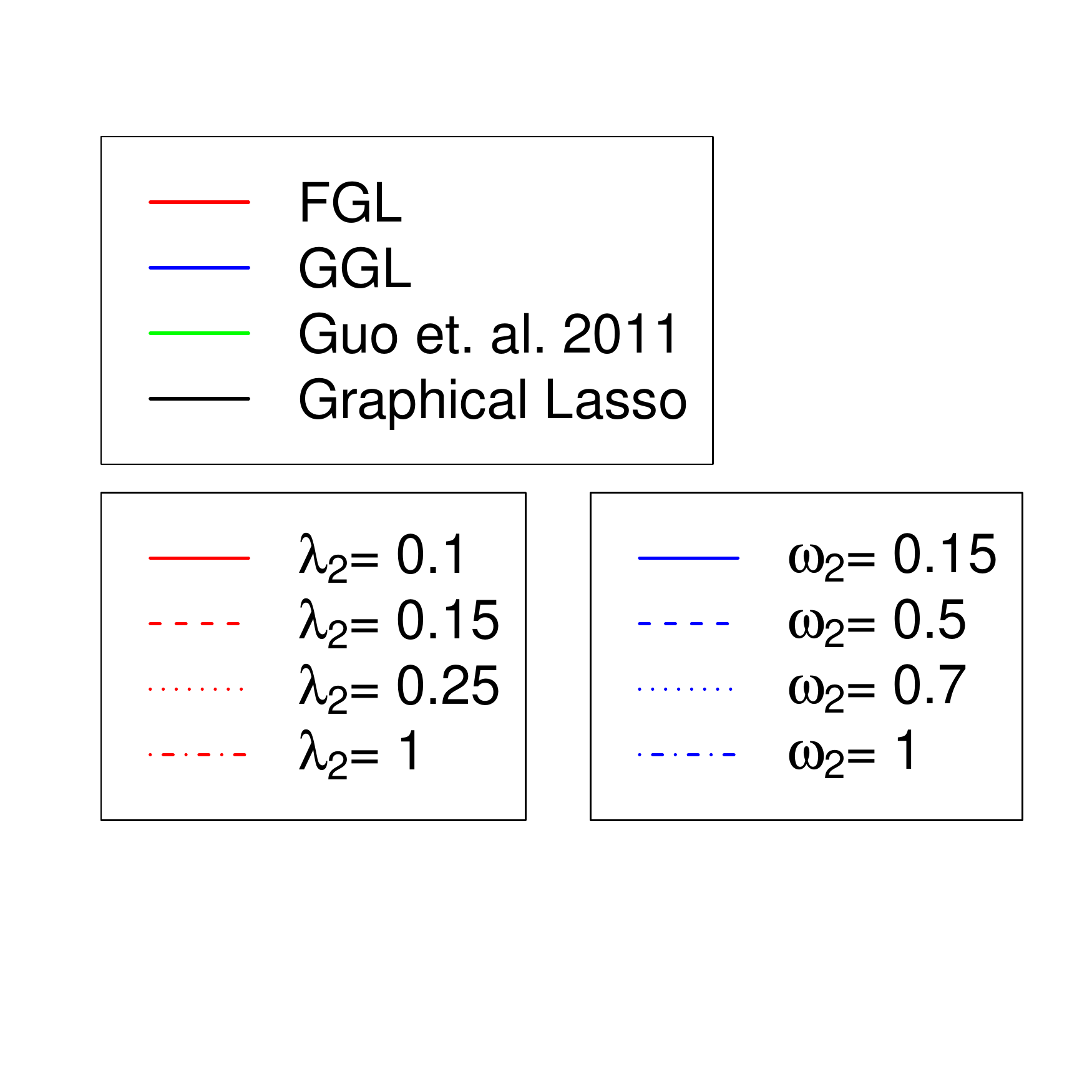}
\caption{\label{Sixplots.K2}
 \it Performance of FGL, GGL, \protect\citet{Guo2011}'s method, and the graphical lasso on simulated data with $150$ observations in each of 2 classes, and $500$ features
 corresponding to ten equally sized unconnected subnetworks drawn from a power law distribution.
 Details are as given in Figure \ref{Sixplots}.}
% {\bf (a)}: The number of edges correctly identified to be nonzero (TP Edges) is plotted against the number of edges incorrectly identified to be nonzero (FP edges).
% {\bf (b)}: The sum of squared errors in edge values is plotted against the total number of edges estimated to be nonzero.
% We calculate the sum of squared errors as $\sum_{k=1}^K \sum_{i\neq j}(\hat{\bTheta}^{(k)}_{ij} - \siginv_{k,ij})^2$.
% {\bf (c)}: The number of edges correctly estimated to have values differing between classes (TP Differential Edges) is plotted against the number of edges incorrectly estimated %to have values differing between classes (FP Differential Edges).
% {\bf (d)}: The Kullback-Leibler divergence of the estimated models from the true models is plotted against the $\ell_1$ norm of the estimated precision matrices.
% {\bf (e)}:  Running time (in seconds) is plotted against the number of non-zero edges estimated. Note the use of a log scale on the $y$-axis.}
\end{figure}

%\subsection*{Simulations with two classes and a single, large power law network}
We also simulated data with an entirely different network structure.
Instead of the block-diagonal network structure used in the previous simulations, in this simulation we generated data
drawn from a single large power law network.
%\subsubsection*{Simulation set-up}
We defined class 1's network to be a single power law network with only one component and 
generated ${\bf \Sigma}_1$ as described in Section \ref{setup}.
We then identified a branch in this network connected to the rest of the network through only one edge.
We then let ${\bf \Sigma}_2^{-1}$ equal ${\bf \Sigma}_1^{-1}$, 
except for the elements corresponding to the edges in the selected branch, which were set to be zero instead.
Finally, we defined ${\bf \Sigma}_2$ by inverting ${\bf \Sigma}_2^{-1}$, and generated the two classes' data using ${\bf \Sigma}_1$ and ${\bf \Sigma}_2$.
This yielded distributions based on two power law networks  that were identical
 except for a missing branch in class 2.
Class 1's network has 499 edges, 104 of which are not present in class 2.
%We generated covariance matrices as described in Section \ref{setup}.
We simulated 100 datasets with $n=150$ observations per class.
Figure \ref{Sixplots.onebignet} shows the results, averaged over the 100 data sets.
Again, FGL and GGL were superior to or competitive with the other methods.  
%
%\subsubsection*{Simulation results}
%
%In each plot, the lines for FGL and for GGL indicate the results obtained with a single value of the similarity tuning parameter, $\lambda_2$ or $\omega_2$.

%Figure \ref{Sixplots.onebignet}a plots the number of edges correctly estimated
%to be nonzero (TP Edges) against the number of edges incorrectly estimated
%to be nonzero (FP Edges).
%Figure \ref{Sixplots.onebignet}b plots the sum of squared errors (SSE) of the estimated edges against
%the total number of edges estimated.  SSE is calculated as $\sum_{k=1}^K \sum_{i\neq j}(\hat{\bTheta}^{(k)}_{ij} - \siginv_{k,ij})^2$.
%Figure \ref{Sixplots.onebignet}c displays the number of edges correctly
%identified as having differing values in the two networks (TP Differential Edges) against
%the number of edges incorrectly identified as having differing values in the
%two networks (FP Differential Edges).
%Figure \ref{Sixplots.onebignet}d displays the the sum of the dKL's of the estimated distributions from the true distributions, as a function of  the $\ell_1$ norm of the estimated %precision matrices.
%Figure \ref{Sixplots.onebignet}e compares the methods' running times in seconds.

\begin{figure}[htp]
\centering
\includegraphics[width=0.25\linewidth]{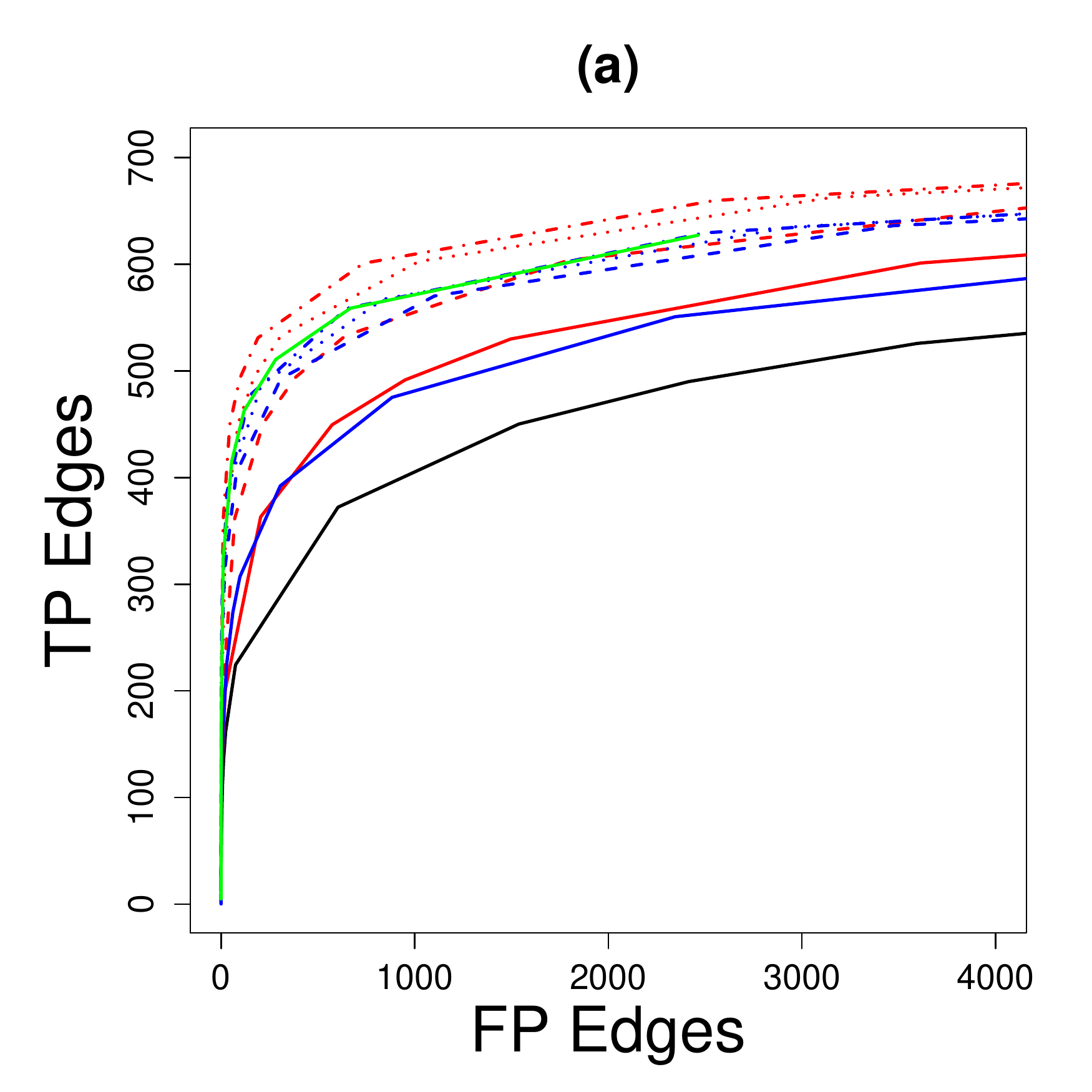} %FPTP! fff  %note: replace these files!
\includegraphics[width=0.25\linewidth]{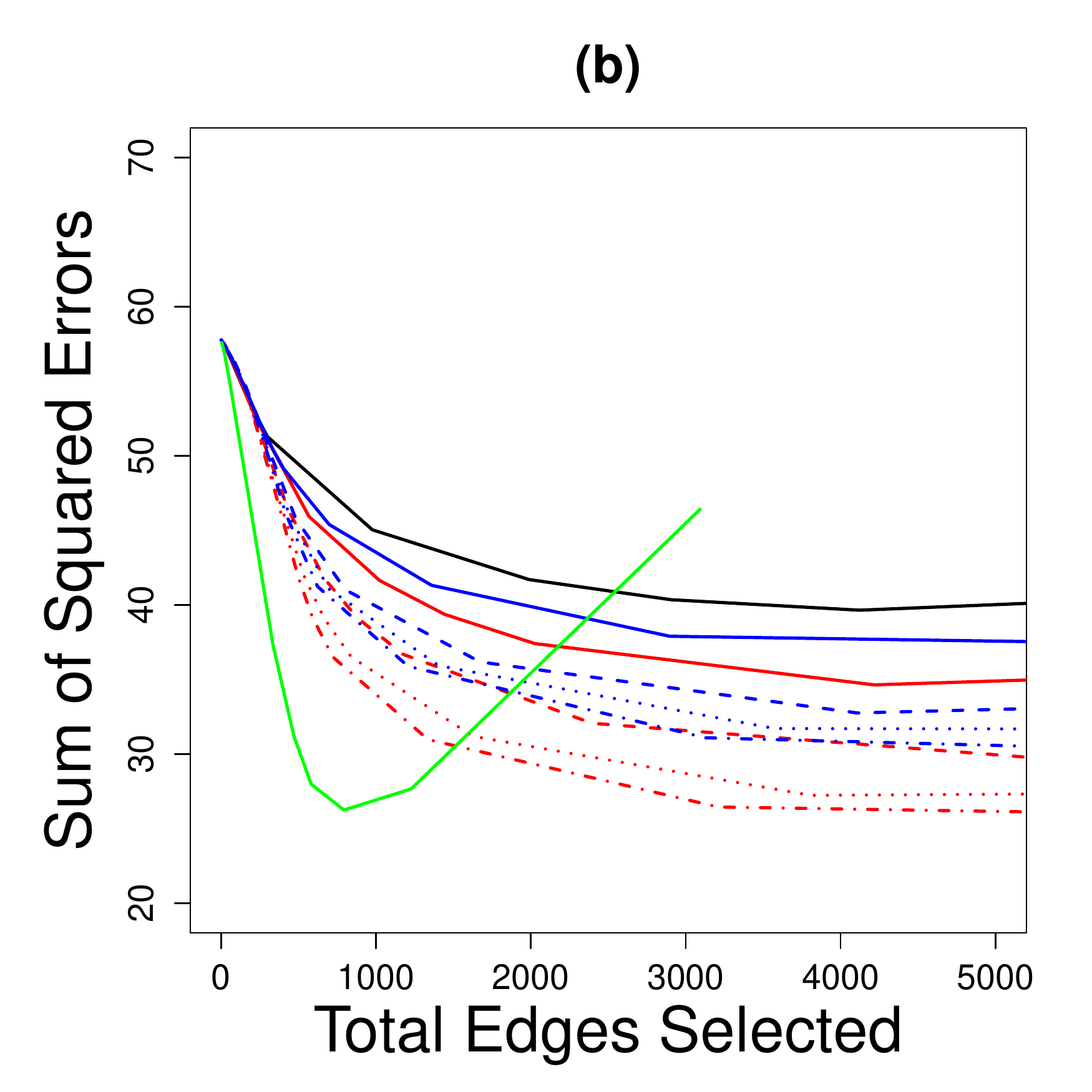}
\includegraphics[width=0.25\linewidth]{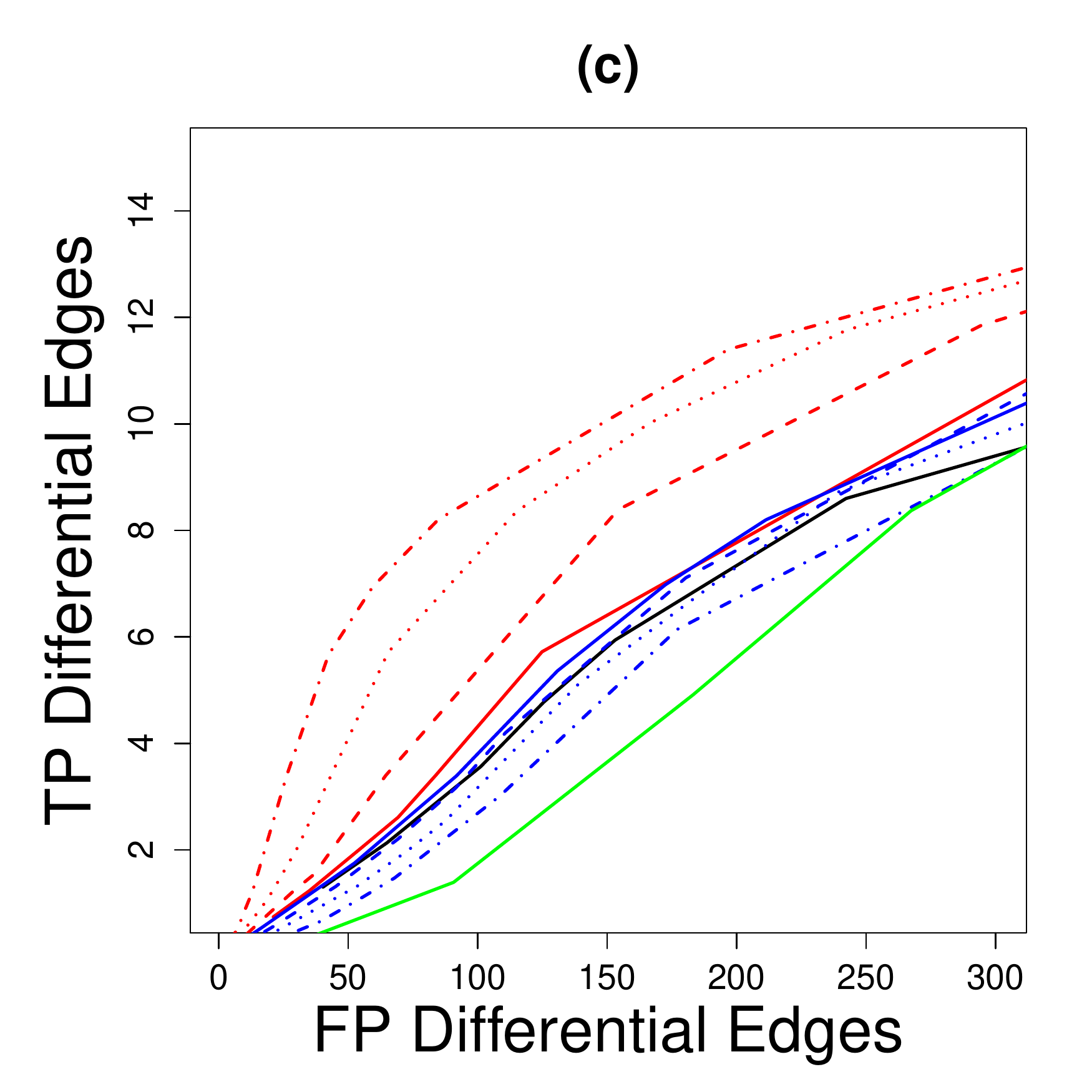}
\includegraphics[width=0.25\linewidth]{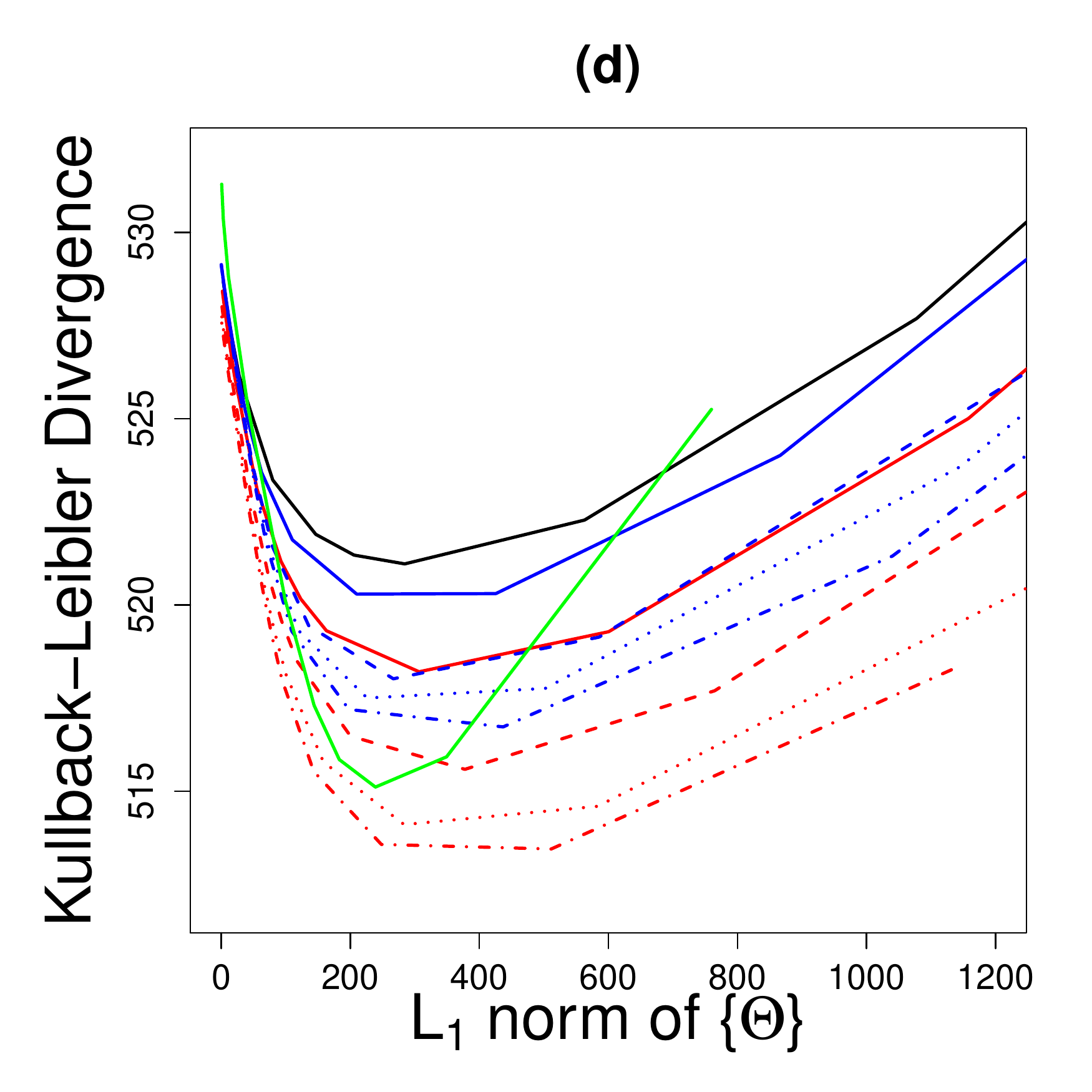}
\includegraphics[width=0.25\linewidth]{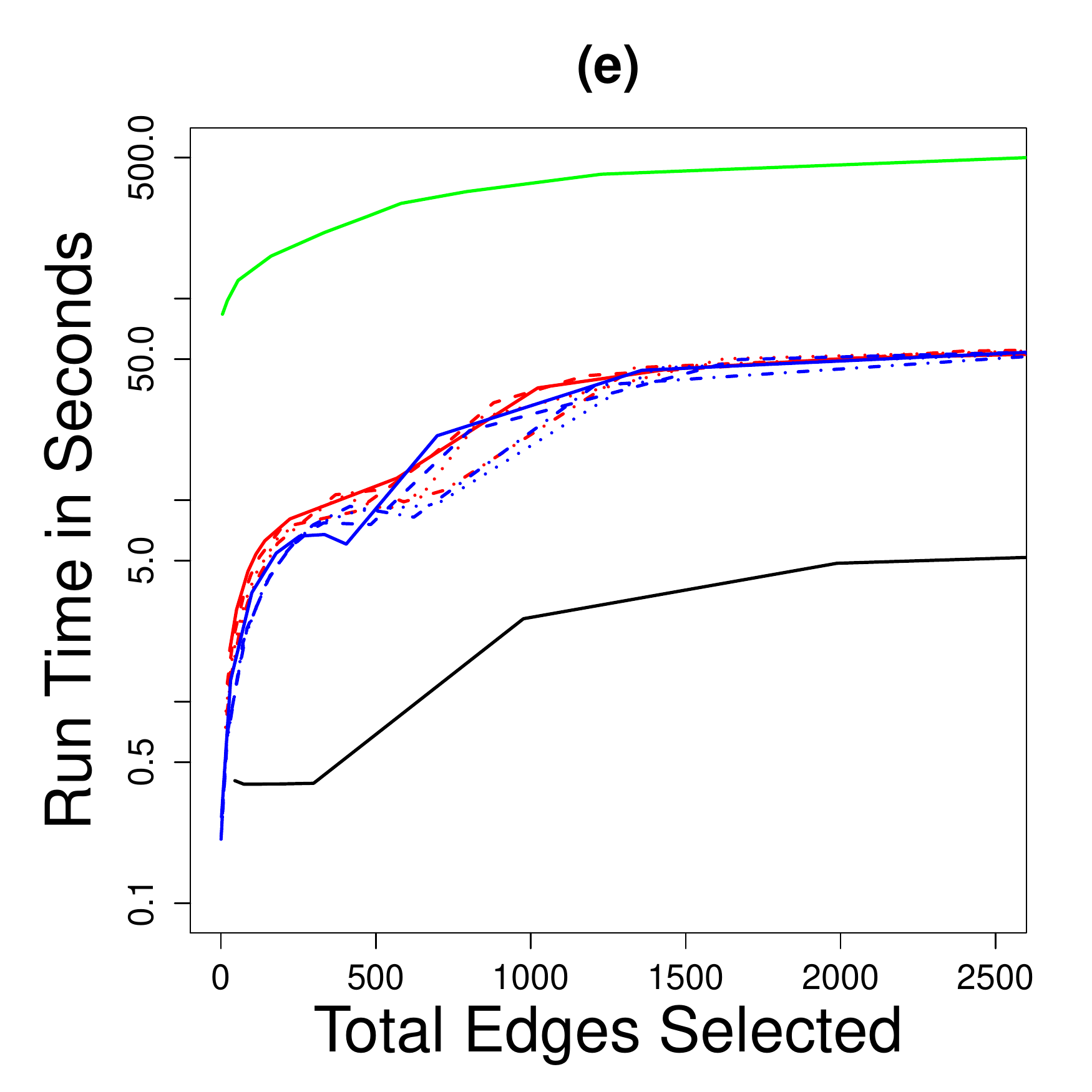}
\includegraphics[width=0.25\linewidth]{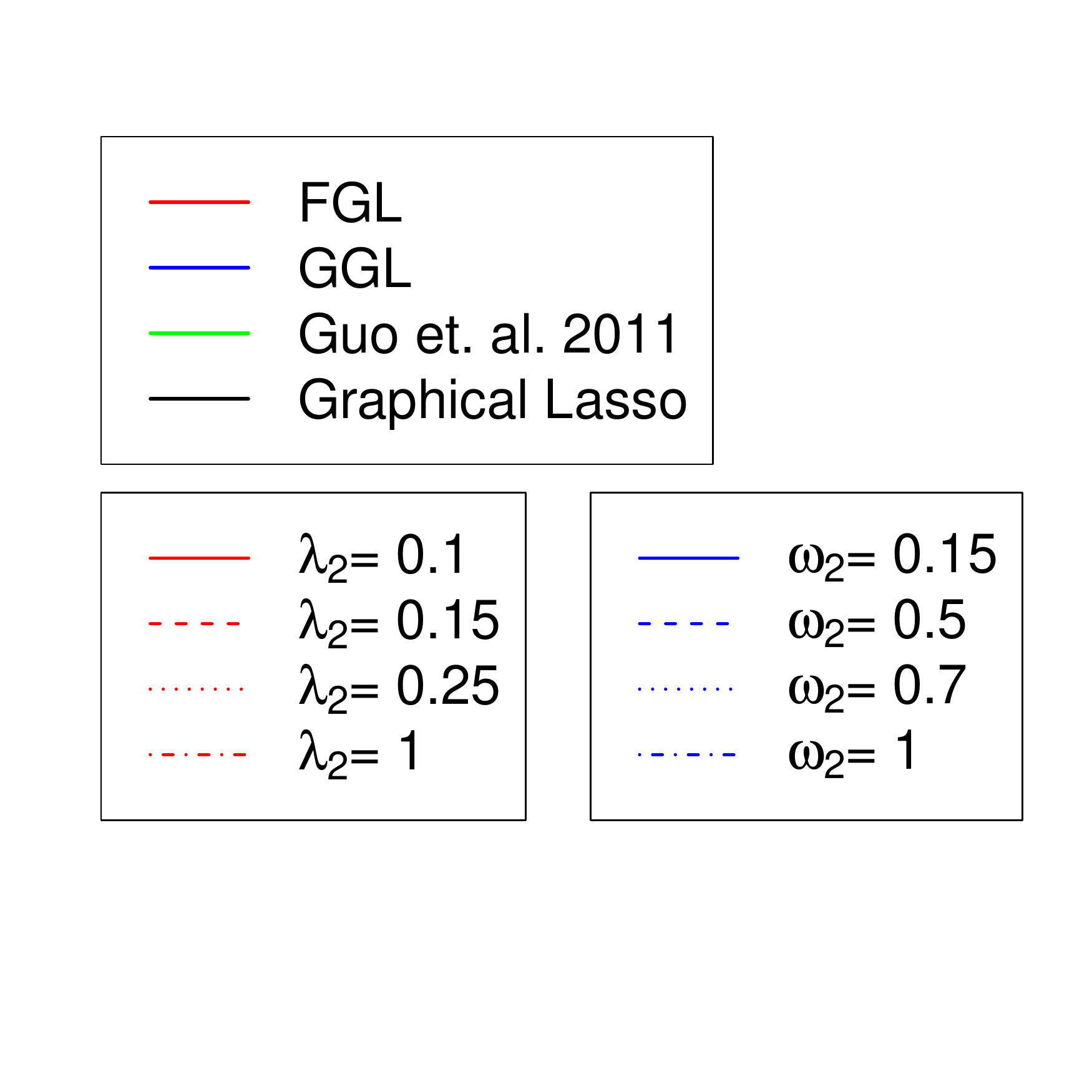}
\caption{\label{Sixplots.onebignet}
 \it Performance of FGL, GGL, \protect\citet{Guo2011}'s method, and the graphical lasso on simulated data with $150$ observations in each of 2 classes, and $500$ features
 corresponding to a single large power law network.
 Details are as given in Figure \ref{Sixplots}.}
% {\bf (a)}: The number of edges correctly identified to be nonzero (TP Edges) is plotted against the number of edges incorrectly identified to be nonzero (FP edges).
% {\bf (b)}: The sum of squared errors in edge values is plotted against the total number of edges estimated to be nonzero.
% We calculate the sum of squared errors as $\sum_{k=1}^K \sum_{i\neq j}(\hat{\bTheta}^{(k)}_{ij} - \siginv_{k,ij})^2$.
% {\bf (c)}: The number of edges correctly estimated to have values differing between classes (TP Differential Edges) is plotted against the number of edges incorrectly estimated %to have values differing between classes (FP Differential Edges).
% {\bf (d)}: The Kullback-Leibler divergence of the estimated models from the true models is plotted against the $\ell_1$ norm of the estimated precision matrices.
% {\bf (e)}:  Running time (in seconds) is plotted against the number of non-zero edges estimated. Note the use of a log scale on the $y$-axis.}
\end{figure}

\section*{Appendix 4: Network structure used in simulations}
\label{simsnet}
The network structure for the simulations described in Section \ref{K3sims} is displayed in Figure \ref{net_used_in_sims}.
Black edges are shared between all three classes' networks, green edges are present only in classes 1 and 2,
and red edges are present only in class 1.

\begin{figure}[htp]
\centering
\includegraphics[width=0.4\linewidth]{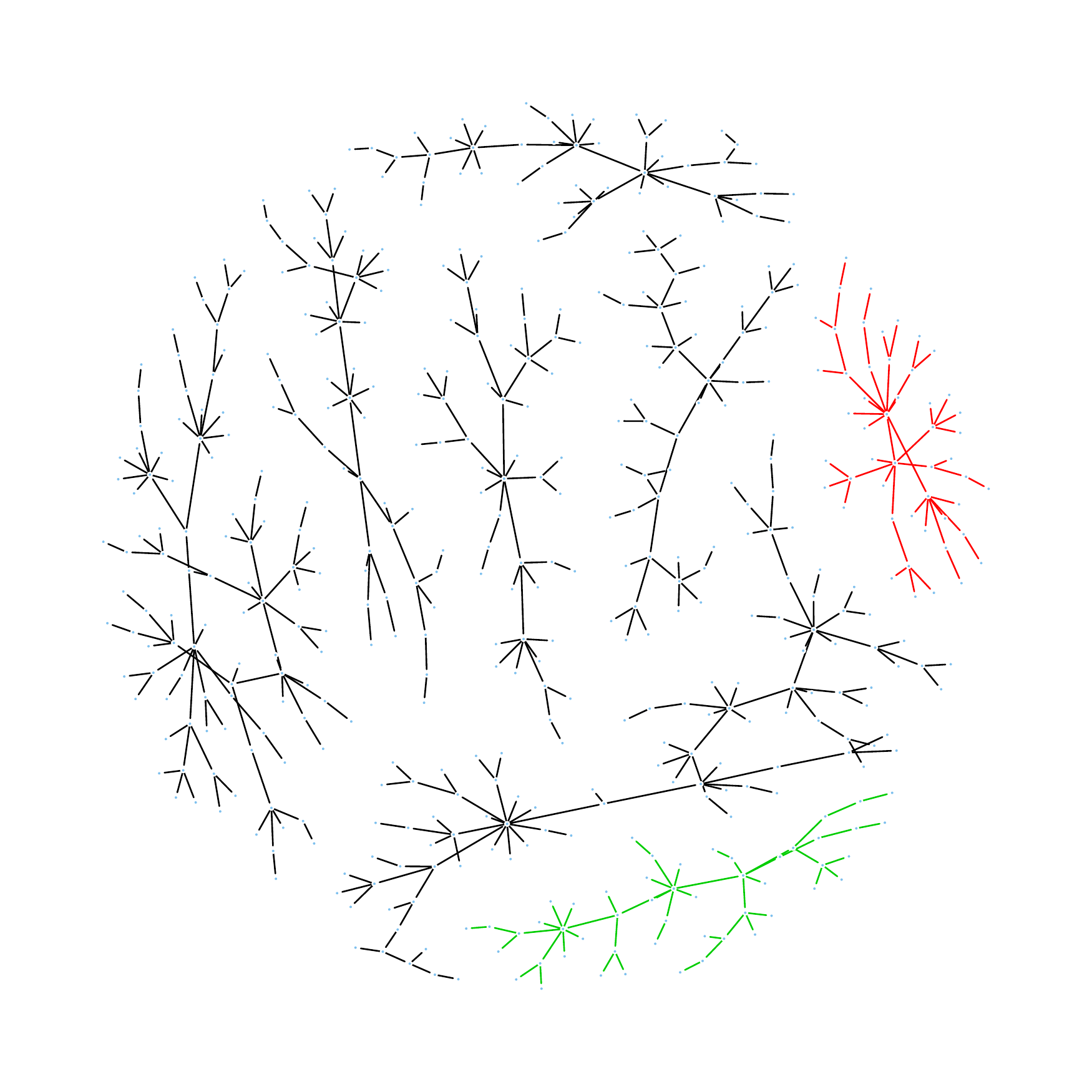}
 \caption{\it{ \label{net_used_in_sims}} \it Network used to generate simulated datasets for Figure \ref{Sixplots} in Section \ref{one}.  Black edges are common to all three classes, green edges are present only in classes 1 and 2,
and red edges are present only in class 1.}
\end{figure}

%\section*{Appendix 6: ADMM algorithm for solving (\ref{simpler}) with the constraint}
%We now consider the problem of solving (\ref{simpler}) subject to the constraint that ${\bf \Theta}^{(k)} \succeq \delta {\bf I} $ for $k=1,\ldots,K$.
%\footnotesize
%\begin{equation}
%\minimize_{\sTheta} \left\{    \frac{1}{2s} \sum_{k=1}^K ||{\bf
%\Theta}^{(k)} - {\bf A}^{(k)} ||_F^2  + P(\sTheta)  \right\} \mbox{ subject to } {\bf \Theta}^{(k)} \succeq \delta {\bf I} \mbox{ for all } k=1,\ldots,K.
%\label{simpler.constraint}
%\end{equation}
%\normalsize
%subject to the constraint that ${\bf \Theta}^{(k)} \succeq \delta {\bf I}$ for $k=1,\ldots,K$.

%The augmented Lagrangian corresponding to this problem takes the form 
%\small
%\begin{equation}
%\label{alm}
%    \frac{1}{2s} \sum_{k=1}^K ||{\bf
%\Theta}^{(k)} - {\bf A}^{(k)} ||_F^2  + P(\sZ)  + \sum_{k=1}^K \trace \left({\bf Y}^{(k)}({\bf \Theta}^{(k)} - {\bf Z}^{(k)}) \right)+ (\rho/2) \sum_{k=1}^K || {\bf \Theta}^{(k)} - {\bf Z}^{(k)}||%_F^2.
%\end{equation}
%Therefore, we can solve (\ref{simpler.constraint}) by iteratively performing the following three steps:
%\begin{enumerate}
%\item Update $\sTheta$ by solving (\ref{alm}) with respect to $\sTheta$. This can be performed by computing $K$ eigendecompositions and replacing any eigenvalues smaller %than $\delta$ with $\delta$.
%\item Update $\sZ$ by solving (\ref{alm}) with respect to $\sZ$.  This is described in Sections~\ref{fgl.ggd.update} and \ref{ggl.ggd.update}.
%\item For $k=1,\ldots,K$, 
%\end{enumerate}
\normalsize

\end{document}